\documentclass[reprint,prl,notitlepage,
superscriptaddress,aps]{revtex4-1}
\usepackage[utf8]{inputenc}
\usepackage{braket}
\usepackage{amsmath}
\usepackage{amsthm}
\usepackage{mathrsfs}
\usepackage{mathtools}
\usepackage{amssymb}
\usepackage{dsfont}
\usepackage{braket}
\usepackage{bm}
\DeclareMathOperator{\tr}{tr}
\newtheorem{thm}{Theorem}[]

\newtheorem{lemma}[thm]{Lemma}
\usepackage{algorithm}
\usepackage{algpseudocode}
\algrenewcommand\algorithmicrequire{\textbf{Input:}}
\algrenewcommand\algorithmicensure{\textbf{Output:}}
\DeclarePairedDelimiter\ceil{\lceil}{\rceil}

\DeclareMathOperator{\arctanh}{arctanh}
\usepackage{soul,xcolor}
\setstcolor{purple}
\begin{document}

\title{Efficient Verification of
 Continuous-Variable Quantum States and Devices without Assuming Identical and Independent Operations}

\author{Ya-Dong Wu}
\affiliation{QICI Quantum Information and Computation Initiative, Department of Computer Science,
The University of Hong Kong, Pokfulam Road, Hong Kong}
\author{Ge Bai}
\affiliation{QICI Quantum Information and Computation Initiative, Department of Computer Science,
The University of Hong Kong, Pokfulam Road, Hong Kong}
\author{Giulio Chiribella}
\affiliation{QICI Quantum Information and Computation Initiative, Department of Computer Science,
The University of Hong Kong, Pokfulam Road, Hong Kong}
\affiliation{The University of Hong Kong Shenzhen Institute of Research and Innovation, 5/F,
Key Laboratory Platform Building, No.6, Yuexing 2nd Rd., Nanshan, Shenzhen 518057, China}
\affiliation{Department of Computer Science, Parks Road, Oxford, OX1 3QD, UK}
\affiliation{Perimeter Institute for Theoretical Physics, Waterloo, Ontario N2L 2Y5, Canada}
\author{Nana Liu}
\affiliation{Institute of Natural Sciences, Shanghai Jiao Tong University, Shanghai 200240, China}
\affiliation{Ministry of Education, Key Laboratory in Scientific and Engineering Computing, Shanghai Jiao Tong University, Shanghai 200240, China}
\affiliation{University of Michigan-Shanghai Jiao Tong University Joint Institute, Shanghai 200240, China}
\begin{abstract}
    Continuous-variable quantum information, encoded into infinite-dimensional quantum systems, is a promising platform for the realization of  many quantum information protocols, including quantum computation, quantum metrology, quantum cryptography, and quantum communication. To successfully demonstrate these protocols, an essential step is the certification of multimode continuous-variable quantum states and quantum devices. This problem is well studied under the assumption that multiple uses of the same device result in identical and independently distributed (i.i.d.) operations. However, in realistic scenarios, identical and independent state preparation and calls to the quantum devices cannot be generally guaranteed. Important instances include adversarial scenarios and instances of time-dependent and correlated noise. In this paper, we propose the first set of reliable protocols for verifying multimode continuous-variable entangled states and devices in these non-i.i.d scenarios. Although not fully universal, these protocols are applicable to Gaussian quantum states, non-Gaussian hypergraph states, as well as amplification, attenuation, and purification of noisy coherent states.
\end{abstract}

\maketitle

 {\em Introduction.~}Continuous-variable (CV) quantum information protocols are widely used in quantum optics~\cite{Bra05,Wee12}. To realise these protocols, it is essential to be able to perform state and device verification on  CV states and devices~\cite{Eis20}. State verification ~\cite{Aol15,Pal18,Tak18,Liu19,Zhu19,Zhu19a,Tak19,Uly19,Uly20} addresses the problem of whether or not a state generated by a quantum device is close enough to a specified target state. While some efficient protocols exist~\cite{Aol15,Liu19}, they require the tested systems to be identically and independently (i.i.d)  prepared, an assumption that is hard to guarantee in realistic scenarios. Quantum device verification~\cite{wu2019efficient} is the problem of determining whether the outputs of a quantum device are close to associated target output states, averaged over all possible input states. CV quantum device verification in the non-i.i.d setting has so far been an open problem. In this paper, we propose verification protocols for multimode CV entangled states and CV quantum devices in non-i.i.d scenarios.

 For finite dimensional systems,   quantum state and quantum device characterization schemes in the non i.i.d setting   have received increasing attention in recent years, motivated by applications in quantum computing and quantum networks with noisy intermediate-scale quantum devices~\cite{Chr12,Pfi18,wallman18,erhard19}. 
There are two important classes of scenarios where the i.i.d assumption cannot be made. The first class includes adversarial scenarios, in which we cannot trust that the adversary will necessarily allow us access to multiple copies of the same  state, or to multiple uses of the same quantum device.  This situation can occur, for instance, in verifiable blind quantum computing~\cite{Ghe19}, where malicious servers can send entangled states to the client to steer computational results. A second class of scenarios involves the presence of time-dependent noise, which can  exhibit correlations between subsequent uses of the same device.  This situation occur, for example, in the transmission of photons through an optical fiber, whose birefringence fluctuates over time \cite{ball2005hybrid}.  In all these cases,  we cannot trust that a realistic quantum device will output identical and independently prepared states in each run.

In the non-i.i.d setting for qubits, a powerful method is to employ the quantum de Finetti theorem, which enables one to approximate a collection of non-i.i.d states by a smaller number of copies of i.i.d states after a randomising procedure followed by tracing out a subsystem~\cite{Mat07}. Leveraging this result, one can reduce the problem of non-i.i.d verification to the i.i.d  scenario. A similar strategy can be used for CV state verification. In the CV setting, there are two main classes of quantum de Finetti theorems, which can be separated into finite dimensional approximations~\cite{Ren09}, and  infinite dimensional constructions~\cite{lev18}. The existing finite dimensional approximations have been developed for applications in quantum key distribution, typically involving single mode systems, and have an exponential scaling of the error in the dimension parameter. On the other hand, in the infinite dimensional constructions one lacks a simple, practically implementable randomising procedure required by the de Finetti theorem to enable the non-i.i.d state to be approximated by i.i.d states.  To circumvent these issues,   we develop a finite dimensional approximation that  can be used for multimode states and has a polynomial scaling of the error with the dimension parameter.

In our approach, we propose a new method, which can be used to verify a broad class of CV quantum states, including multimode Gaussian states and CV hypergraph states. Unlike previous approaches, which used permutation symmetry by randomly reshuffling the various systems, our test takes advantage of an additional symmetry property, namely symmetry with respect to rotations in phase space~\cite{Lev13}. This additional symmetry allows us to overcome all the challenges of the non-i.i.d. setting. In our protocol, the initial non-i.i.d state is randomized not only by a random permutation but also by random phase rotations at each subsystem. These rotations can be performed without loss of generality owing to the symmetry of the states under consideration. Exploiting this rotational symmetry, we are able to achieve polynomial scaling of the approximation error between the randomised non-i.i.d state and its i.i.d approximation with respect to an effective finite  dimension $d$ associated to the family of states under consideration.  

Building on our i.i.d approximation, we then construct a verification protocol with the desirable properties of soundness and completeness, which are necessary for successful verification. Soundness of a protocol means that the probability of false positives is low: if the actual state is orthogonal to the target state, it should have a low probability of passing the verification test. Completeness means that the correct state has a high probability to pass the test.  Thanks to rotational symmetry, we show that the complexity of our verification protocol has a favourable scaling in terms of the soundness and completeness parameters. 


Building on our CV verification results, we also provide the first protocol for CV non-i.i.d quantum device verification. This protocol combines a duality between state tests and channel tests introduced in Ref.~\cite{bai2018test} and our new techniques in CV state verification. With these ingredients, we can demonstrate bounds on the completeness and soundness of device verification.   

 {\em Framework.~}  We now introduce the necessary basics of CV quantum states and the task of verification, before going on to demonstrate explicitly our protocols for specific classes of CV states and channels. 

A CV state lies on an infinite dimensional Hilbert space, equipped with observables with a  continuous spectrum, such as the position and momentum observables of a quantum particle. CV states are usually implemented by bosonic systems, described by quantum harmonic oscillators. 
CV quantum information is encoded in the tensor product $\mathcal H^{\otimes k}$ of Hilbert space $\mathcal H=\text{Span}\{\ket{n}\}_{n\in \mathbb N}$, where $\hat{n}\ket{n}=n\ket{n}$ is a particle number eigenstate with particle number operator $\hat{n}=\hat{a}^\dagger\hat{a}$.
    Quadrature operators are $\hat{q}:=\frac{\hat{a}+\hat{a}^\dagger}{\sqrt{2}}$ and $\hat{p}:=\frac{\hat{a}-\hat{a}^\dagger}{\sqrt{2} \text{i}}$.
    For $k$-mode CV states, the quadrature operators are denoted by  vector $\hat{\bm{x}}:=(\hat{q}_1,\hat{p}_1,\dots, \hat{q}_k,\hat{p}_k)$.
    
    An important class of CV states are Gaussian states. Pure Gaussian states can be written in the form $U_{\bm{S},\bm{d}}\ket{0}^{\otimes k}$, 
    where  $U_{\bm{S},\bm{d}}$ is a Gaussian unitary operation, characterized by an affine mapping $(\bm{S},\bm{d}):\hat{\bm{x}}\rightarrow \bm{S}\hat{\bm{x}}+\bm{d}$, 
    where $\bm{S}\in \mathbb{R}^{2k\times 2k}$ is a symplectic transformation and $\bm{d}\in\mathbb R^{2k}$. 
    
    
 The most common CV measurement is homodyne detection~\cite{bachor2004}, which is routinely implemented in quantum optics laboratories.   Mathematically, the homodyne measurement corresponds to a projective measurement of a quadrature operator. This means that the expectation value of 
any linear combination of quadratures  $\hat{q}(\theta):=\cos\theta \, \hat{q}+\sin\theta \, \hat{p}$ and $\hat{p}(\theta):=-\sin\theta\, \hat{q}+\cos\theta\, \hat{p}$, with $\theta\in[0,\pi/2)$, can be measured using homodyne detection in a rotated basis.

 In state verification, a verifier has to  test the preparation of  a target state, denoted  by $\ket{\phi}\in \mathcal H^{\otimes k}$, where $k\in\mathbb{N}^+$.   The verifier is given $n$ quantum registers, whose state is claimed to consist of  $n$ identical copies of the target state.  The actual state of the $n$ registers is unknown to the verifier, and is denoted by $\rho^{(n)} \in \mathcal S( \mathcal H^{\otimes k\cdot n})$. The state $\rho^{(n)}$ could deviate from the ideal state $\ket{\phi}\bra{\phi}^{\otimes n}$ due to imperfections of the source, or could even be prepared by a potentially malicious server.
The verifier then chooses $n-m$ quantum registers uniformly at random, and performs 
 measurements on each register, to decide whether the reduced state at the remaining $m$ registers is close enough to $\ket{\phi}\bra{\phi}^{\otimes m}$ or not. From now on, we use the term randomly choosing to mean choosing from an uniform random distribution.
Denoting~$0\le T\le \mathds{1}$ as the POVM element on $\mathcal H^{\otimes k(n-m)}$ that corresponds to the verification test flagged as passed, and 
 $0<\epsilon_s, \epsilon_c <\frac{1}{2}$ as failure probabilities, a reliable quantum state verification scheme must satisfy
 \begin{itemize}
 \item soundness: for any permutation-invariant $\rho \in \mathcal S( \mathcal H^{\otimes k\cdot n})$,
$ \tr\left(T\otimes (\mathds{1}-\ket{\phi}\bra{\phi}^{\otimes m})\rho\right)\le \epsilon_s$,
 and 
 \item completeness:
$\tr\left(T\ket{\phi}\bra{\phi}^{\otimes (n-m)}\right)\ge 1-\epsilon_c$.
 \end{itemize}
 Intuitively, a good bound on soundness denotes a low probability of a false positive, that is, a low joint probability that the test is passed and yet the remaining state is orthogonal to the target state. On the other hand, completeness guarantees that if the state is identical to the target state, it must pass the verification test with a high probability.

The task of quantum device verification, closely related to state verification, is to determine whether the outputs of a quantum device are close to target output states or not, when averaged over a fiducial ensemble of input states.
 We can define an ensemble of input states as $\{p_x, \rho_x\}_{x\in X}$, where $X$ is an index set, $\{p_x\}_{x\in X}$ is a probability distribution, and $\rho_x\in \mathcal S(\mathcal H^{\otimes k})$. Suppose the target outputs are pure states $\{\ket{\phi}_x\}_{x\in X}$, where $\ket{\phi}_x\in \mathcal H^{\otimes k}$. A target channel~${\mathcal E}_{\rm t}$ is defined as the quantum channel that achieves the maximal average fidelity
$\bar{F}(\mathcal E):=\sum_{x\in X}p_x \braket{\phi_x|\mathcal E(\rho_x)|\phi_x}$,
and its maximum achievable value is denoted by $\bar{F}_{\max}$~\cite{chiribella2013optimal,yang14}. 
 
In this context, an important observation is that any test of quantum devices can be realized by preparing a single entangled state on the input and an ancillary system, and to perform a single joint measurement on the output and the ancillary system~\cite{bai2018test}.  This observation yields a general device verification protocol similar to state verification above.   Let  $\mathcal E^{(n)}$ be an $n\cdot k$-mode quantum channel, claimed to  act as  $n$ independent uses of the $k$-mode target channel ${\mathcal E}_{\rm t}$. Here we regard ${\cal E}^{(n)}$ as a channel with $n$ inputs, each input consisting of $k$ modes.  
 The verifier then randomly chooses $(n-m)$  inputs and injects one part of a bipartite entangled state into each of these inputs. Then, the verifier can apply local measurements at the outputs and the ancillary systems, to determine whether the channel $\mathcal E^{(m)}$ at the remaining $m$ inputs is close to $\mathcal E_t^{\otimes m}$ or not.
 
A reliable device verification scheme must similarly satisfy soundness and completeness conditions
\begin{itemize}
    \item  soundness: for any permutation-invariant $n$-input channel $\mathcal E^{(n)}$, 
    \begin{equation}\label{channelsoundness}
    \left[ T\otimes \left(\mathds 1 -\frac{\bar{F}^{\otimes m}}{\bar{F}_{\max}^m}\right)\right]\left(\mathcal E^{(n)}\right)\le \epsilon_s
    \end{equation}
    where $T$ is the map from an $(n-m)$-input quantum channel to the probability of passing the test, and $\mathds 1$ is a map that maps any $m$-input channel into the number $1$.  
    \item completeness: 
    \begin{equation}\label{channelcompleteness}
       T\left(\mathcal E_t^{\otimes (n-m)}\right)\ge 1-\epsilon_c.
    \end{equation}
\end{itemize}
The soundness of channel verification is analogous to that of state verification, except here the figure of merit is average fidelity instead of fidelity between the prepared state and the target state.
    
 State verification under the i.i.d assumption can be performed by detecting a fidelity witness $W$, which is  an observable whose expectation value with respect to any prepared state is a tight lower bound of its fidelity with the target state. This provides an efficient approach to verify both CV quantum states~\cite{Aol15, Liu19} and CV quantum channels~\cite{wu2019efficient}. 
 In this paper, although we do not have the i.i.d assumption, we will continue to use these techniques after obtaining an i.i.d approximation. 
  
To obtain an i.i.d approximation using a finite~$d$ de Finetti theorem, one needs to filter CV states so they effectively lie on a $d$-dimensional subspace. We note that although one cannot infer whether all the remaining subsystems are bounded to lie on a finite $d$-dimensional subspace by testing partial subsystems, it is possible to deduce whether a CV state is bounded for most subsystems. 
Then through randomization in terms of both permutation and phase rotations, this almost-bounded CV state is then close to an i.i.d $d$-dimensional state, after tracing out part of its subsystems.

In general non-i.i.d settings, CV quantum state verification comprises of two subprotocols: the dimension test and the fidelity test.
The dimension test is used to bound the dimension $d$.
In the dimension test, the measurement outcomes of homodyne detection are compared with a certain threshold. If the measurement outcomes are always less than the threshold, this gives a strong guarantee that each subsystem is confined in a subspace spanned by Fock states $\ket{n}$ with $n$ less than $d$. 
  Through discarding a large fraction of the subsystems of the randomized non-i.i.d state, one can treat the state at the remaining subsystems as approximately i.i.d, due to a finite-$d$ de Finetti theorem. After getting an i.i.d approximation, the fidelity test, similar to the test under i.i.d assumption, is to certify the fidelity between the state at each remaining subsystem and the target state, by detecting the fidelity witness at partial subsystems. Figure~\ref{scheme} summarises the key steps of the scheme.
   
  {\em The verification protocol.}~ Suppose the target state $\ket{\phi}=U \ket{0}^{\otimes k}$ is a multimode entangled state,  mathematically obtained by applying a suitable unitary operator $U$ to the vacuum. Given $(k/2+1)N$ quantum registers, each of which stores a $k$-mode quantum state, the verifier uses $kN/2$ registers for the dimension test. Here $N$ is chosen to be an even integer. The first step of the dimension test is to divide the $kN/2$ registers into $k$ groups of $N/2$ registers each. In each group, by comparing the square of homodyne detection outcomes with an upper bound $d_0/2>0$ for~$N/2$ registers, the verifier infers whether most of the $k$-mode states in the remaining $N$ registers fall on a finite-dimensional subspace 
$\bar{\mathcal H}_j:=\left\{U \ket{n_i}^{\otimes_{i=1}^k}| \forall i, n_{i}\in \mathbb{N}, n_j<d_0 \right\}$  of $\mathcal H^{\otimes k}$, where $ 1\le j\le k$.
If the $k$ groups all pass the test, then most subsystems at the remaining $N$ registers fall on a finite-dimensional subspace
$\bar{\mathcal H}:=\cap_{j=1}^k \bar{\mathcal H}_j$
of $\mathcal H^{\otimes k}$.
Then after discarding a large fraction of the remaining $N$ registers and keeping only $L$ registers $(L\ll N)$, the reduced state $\rho^{(L)}$ at the remaining $L$ registers can be shown to
fall on $\bar{\mathcal H}^{\otimes L}$ and is approximately i.i.d to high probability. Proofs of these statements can be found in the supplemental Material  \footnote{See the supplemental material}. Finally, the verifier chooses $L-m$ of the remaining registers to perform the fidelity test. Here one estimates
the expectation value of chosen fidelity witness $\mathds{1}-U\hat{n}U^\dagger$ at $L-m$ registers. The outcome of the fidelity then determines whether the fidelity between the states at the remaining $m$ registers and the tensor product of $m$ target states is close to one. We will later explain the detailed procedure of the dimension test and the fidelity test for specific target states.

 \begin{figure}
\includegraphics[width=0.48\textwidth]{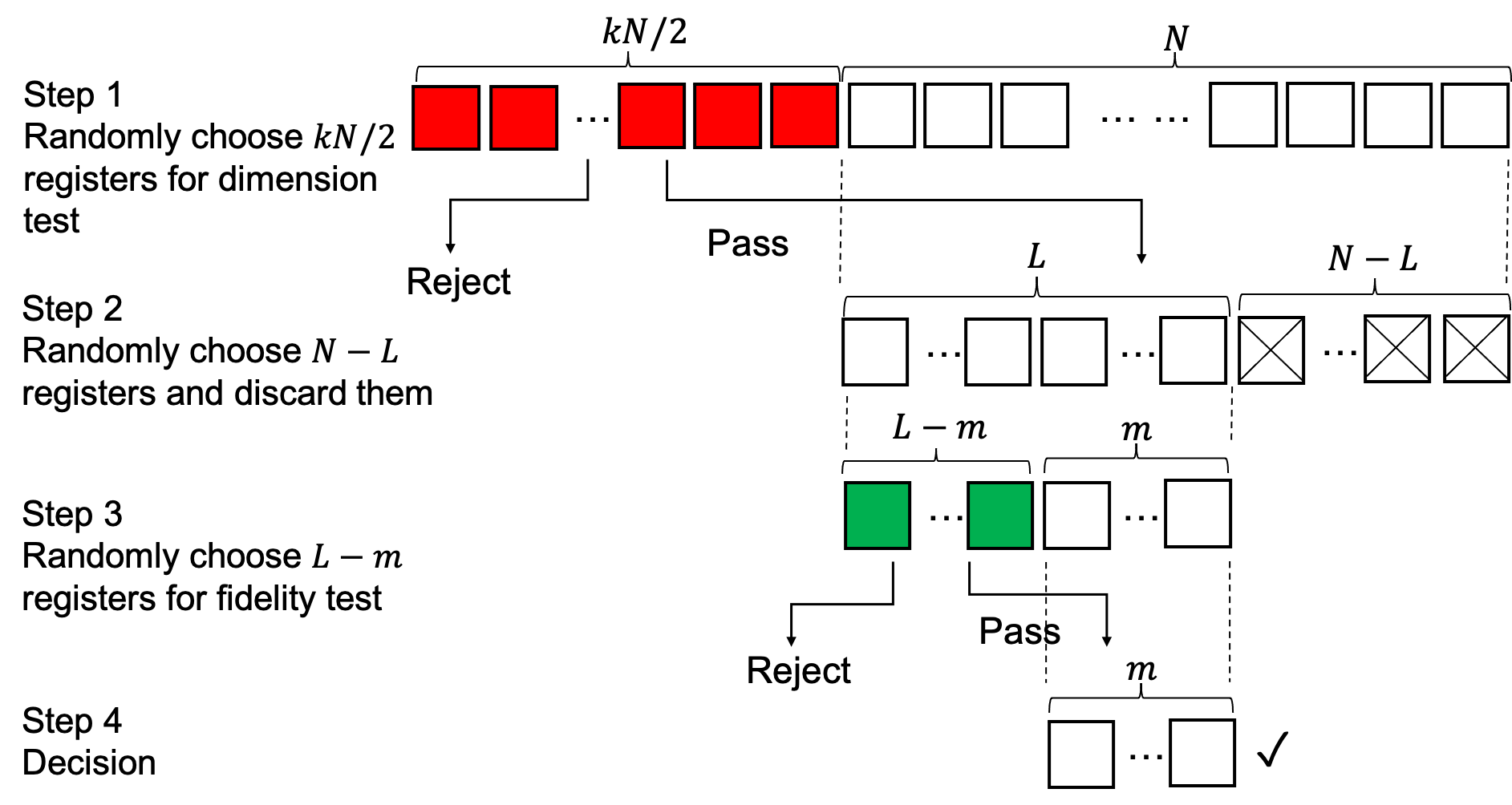}
\caption{The verifier receives $(k/2+1)N$ registers, each of which is represented by a box and contains an unknown $k$-mode state. The verifier randomly chooses $kN/2$ registers, represented by red boxes, to apply a dimension test. If the dimension test is passed, then the verifier goes on with the fidelity test at the other $N$ registers. Otherwise, the verifier aborts the test (rejects). Suppose the dimension test is passed. Then, the verifier randomly chooses $N-L$ registers, represented by boxes with crosses inside, and discards  them. For the remaining registers, the verifier randomly chooses $L-m$ registers, represented by green boxes, to perform the fidelity test. If the fidelity test is passed, then the verifier takes the state in the remaining $m$ registers, represented by blank boxes, as reliable copies of target state $\ket{\psi}$. Otherwise, the verifier rejects the remaining states. }\label{scheme}
\end{figure}

At each round of testing, each register is randomly chosen and this randomisation guarantees permutation invariance of the registers.
Besides permutational symmetry, our verification test also exhibits  an additional symmetry, owing to the fact that the vacuum state $|0\rangle^{\otimes k}  =  U^\dag |\phi\rangle$ is invariant under rotations in phase space.  This additional symmetry is enforced by first applying the unitary operation $U^\dagger$, and then applying a homodyne detection in a randomly rotated quadrature basis at each mode.  Practically, for certain unitaries~$U$ like Gaussian unitary operations, the    the application of the unitary gate $U^\dagger$ can be omitted, because it can be reproduced by classical processing of the measurement outcomes. 
Because of this rotational symmetry, only the diagonal entries of any $\rho^{(kN/2+N)}\in \mathcal L(\mathcal H^{\otimes k(kN/2+N)})$ in the basis $\{U\ket{n_i}^{\otimes_{i=1}^k}\}^{\otimes (kN/2+N)}$ affects the results of this test. 

Now we describe the dimension test in detail for pure Gaussian target states.
In the dimension test, the verifier divides $kN/2$ registers into $k$ groups. 
In $j$th  group $(j\in [k])$, the verifier randomly chooses phase $\theta_l\in[0,\frac{\pi}{2})$ ($l\in [N/2]$) at each register and
 measures either $\hat{\tilde{q}}_j(\theta_l)$ or 
$\hat{\tilde{p}}_j(\theta_l)$, where $\hat{\tilde{q}}_j=\sum_{1\le i\le 2k}\bm{S}_{2j-1, i}^\top(\hat{\bm{x}}_i-\bm{d}_i)$, and $\hat{\tilde{p}}_j=\sum_{1\le i\le 2k}\bm{S}_{2j, i}^\top(\hat{\bm{x}}_i-\bm{d}_i)$ are both linear combinations of local quadrature operators. 
Repeat the measurement in each group for $N/2$ times, and denote the $l$th measurement outcome in the $j$th group by $f_{j,l}$. 
For each measurement outcome, the verifier defines an associated variable $z_{j,l}$: if $(f_{j,l})^2>  d_0/2$, $z_{j,l}=1$; otherwise, $z_{j,l}=0$. 
After homodyne measurements on the $kN/2$ registers, if for all~$j\in [k]$, $\sum_{l=1}^{N/2} z_{j,l} \le N\text{e}^{-c_0^2 d_0}$, with $c_0=1-\frac{1}{\sqrt{2}}$, then the states are considered to have passed the dimension test; otherwise, the verifier aborts the test as soon as any $j$ fails the $\sum_{l=1}^{N/2} z_{j,l} \le N\text{e}^{-c_0^2 d_0}$ condition and rejects all the states.

If the states pass the dimension test, then the verifier randomly chooses
$ L=\ceil*{264k^2 m^2 d_0^2\ln \frac{4}{\epsilon}/\epsilon^2+m}$
 from the remaining registers, where $0<\epsilon<1/2$ is a tolerant failure probability, and discards all the other registers. These registers are now used for the fidelity test, where the verifier first randomly chooses $L-m$ registers from the $L$ registers. 
At the $i$th register ($i\in [L-m]$), the verifier then randomly chooses $j_i\in [k]$ and $\theta_i\in[0,\frac{\pi}{2})$, and measures either $\hat{\tilde{q}}_{j_i}(\theta_i)$ or $\hat{\tilde{p}}_{j_i}(\theta_i)$ randomly. Denote the measurement outcome by $\chi_i$.
After $L-m$ rounds of measurements, the verifier compares an estimator of the fidelity witness 
$
W^*=1+k/2 -k/(L-m)\sum_{i=1}^{L-m} \chi_i^2
$
with threshold $1-\frac{\epsilon}{2m}$.
If $W^*\ge 1-\frac{\epsilon}{2m}$, the verifier accepts the state at remaining~$m$ registers as reliable copies of $\ket{\phi}$. Otherwise, the verifier rejects.

This scheme also works for verification of non-Gaussian CV hypergraph states~\cite{Tak19,Moo19}, 
where the verifier follows the same procedure, except that  $\hat{\tilde{q}}_j:=U\hat{q}_j U^\dagger$ and $\hat{\tilde{p}}_j:=U\hat{p}_j U^\dagger$ are different from above.

\noindent
\textbf{Theorem.}
Suppose $\ket{\phi}=U\ket{0}^{\otimes k}$ is a $k$-mode entangled state, where $U$ satisfies that $U$ followed by homodyne detections can be simulated by homodyne detections followed by classical processing of measurement outcomes. 
When 
\begin{equation}\label{d0requirement}
    10k\text{e}^{-c_0^2d_0}\left(\frac{264k^2 m^2 d_0^2\ln \frac{4}{\epsilon}}{\epsilon^2}+m\right)\le \epsilon,
\end{equation}
and 
\begin{equation}\label{Nrequirement}
N> \frac{50}{64}\ln\frac{4k}{\epsilon}\text{e}^{2c_0^2d_0},
\end{equation} 
this verification scheme, characterized by $T$, satisfies soundness, i.e.,  for any permutation-invariant $\rho \in \mathcal S( \mathcal H^{\otimes k\cdot n})$, 
$\tr\left(T\otimes (\mathds{1}-\ket{\phi}\bra{\phi}^{\otimes m})\rho\right)\le \epsilon$,
 and completeness, i.e.,
$ \tr\left(T\ket{\phi}\bra{\phi}^{\otimes ((k/2+1)N-m)}\right)\ge 1-\epsilon$.

The sample complexity of this verification scheme is 
\begin{equation}\label{sampleComplexity}
    (k/2+1)N=O\left(\frac{k^7m^4}{\epsilon^6} \text{Poly}\left(\ln \frac{k m}{\epsilon}\right)\right),
\end{equation} 
 which can be considered as a theoretical upper bound of the minimum required samples for the most general scenarios without energy cutoffs.
Compared to the sample complexity $L=O\left(\frac{k^2 m^2 }{\epsilon^2} \text{Poly}\left(\ln \frac{k m}{\epsilon}\right)\right)$ in the i.i.d case, 
at most $L^4$ samples are sufficient for CV-state verification in non-i.i.d scenario.
In experiments, unknown quantum states can be sent to the verifier through light pulses, and the verifier implements the verification test by applying homodyne detections on the sequence of pulses. 
If we assume that each mode is confined in a subspace spanned by Fock states $\ket{n}$ with $n< d_0$, then the sample complexity is reduced to $O\left(\frac{k^4m^2d_0^4\ln1/\epsilon}{\epsilon^3}\right)$. Using state-of-the-art homodyne detections~\cite{shaked2018,takeda2019}, for $k^4d_0^4\lesssim \frac{10^{13}\epsilon^3}{m^2\ln1/\epsilon}$, the verification test can be accomplished within a few hours.

These same state verification techniques can also be used to implement the verification of quantum devices. We begin with the observation that any test of quantum devices can be realized by preparing one entangled state on the input and an ancillary system, and then jointly measuring the output and the ancillary system~\cite{bai2018test}. The observable to be measured can then be chosen to be (average) fidelity witness as in a state verification task~\cite{wu2019efficient}.
By adding a dimension test and rotational symmetry in the fidelity test, we get our quantum-device verification schemes.  
Verification protocols of amplification, attenuation, and purification of noisy coherent states can be found in the supplemental material.

\noindent
\textbf{Corollary.}
Suppose the target device $\mathcal E_t$ is an optimal quantum device for amplification, attenuation or purification of noisy coherent states, or a unitary $U$ satisfying that $U$ followed by homodyne detections can be simulated by homodyne detections followed by classical processing of measurement outcomes, and the ensemble state of input is a Gaussian state. Then 
when $d_0$ and $N$ satisfy Eqs.~(\ref{d0requirement}) and (\ref{Nrequirement}), respectively,
the verification scheme satisfies soundness (\ref{channelsoundness})  and completeness (\ref{channelcompleteness}) with $n=(k/2+1)N$ and $\epsilon_s=\epsilon_c=\epsilon$.

A verification scheme of $k$-mode quantum devices has the same sample complexity as shown in Eq.(\ref{sampleComplexity}).

{\em Conclusions.~}We have proposed the first protocols that can verify both multimode CV entangled states and CV quantum devices without the assumption of i.i.d state and device preparation and bounded statistical moments of quadratures.  Through bypassing the i.i.d assumption for multimode states, our results can be applied to CV blind quantum computing~\cite{Mor12,marshall16,Liu19}, where a potentially malicious server may deceive an agent or steer the computational results by preparing entangled states.
 Our results can also be applied to
performance benchmarks of quantum devices~\cite{braunstein2000,hammerer2005,Ryo08,owari2008,Gerardo08,chiribella2013optimal,Giulio14,yang14,bai2018test}, in a broader setting where the  devices may undergo arbitrary correlated noise processes in subsequent uses, and may contain an internal  memory that affects their behavior  on later inputs.

 {\em Acknowledgements.~} The authors are grateful to Barry C. Sanders (University of Calgary), Carlos Navarrete-Benlloch (Shanghai Jiao Tong University) and Huangjun Zhu (Fudan University) for interesting and fruitful discussions. 
YDW, GB and GC acknowledge funding from the National Natural Science Foundation of China grant no.\ 11675136,
and the Hong Kong Research Grant Council grants no.\ 17300918 and no.\ 17307520. 
NL acknowledges funding from the Shanghai Pujiang Talent Grant (no. 20PJ1408400) and the NSFC International Young Scientists Project (no. 12050410230).  NL is also supported by the Innovation Program of the Shanghai Municipal Education Commission (no. 2021-01-07-00-02-E00087), the Shanghai Municipal Science and Technology Major Project (2021SHZDZX0102) and the Natural Science Foundation of Shanghai grant 21ZR1431000.

\newpage

\begin{widetext}

\section{Verification of CV quantum states}

In the following, we present verification protocols for Gaussian states and CV hypergraph states, respectively.

 \subsubsection{Verification of Gaussian pure states}
  Any Gaussian pure state can be written as 
 \begin{equation}
 \ket{\psi}=U_{\bm{S}, \bm{d}}\ket{0},
 \end{equation}
 where $U_{\bm{S}, \bm{d}}$ is a Gaussian unitary operation, which yields an affine transformation $\hat{\bm{x}}\rightarrow \bm{S}\hat{\bm{x}}+\bm{d}$ with $\hat{\bm{x}}=(\hat{q}_1, \hat{p}_1, \dots, \hat{q}_k, \hat{p}_k)^\top$.
 
\begin{enumerate}
\item dimension test: divide $K:=kN/2$ registers evenly into $k$ groups. At $i$th register of $j$th group, 
randomly choose to measure either $\hat{\tilde{q}}_j(\theta):=\cos\theta_j^i \sum_{1\le l\le 2k}\bm{S}_{2j-1, l}^\top(\hat{\bm{x}}_l-\bm{d}_l) +\sin\theta_j^i \sum_{1\le l\le 2k}\bm{S}_{2j, l}^\top(\hat{\bm{x}}_l-\bm{d}_l)$ or 
$\hat{\tilde{p}}_j(\theta):=-\sin\theta_j^i \sum_{1\le l\le 2k}\bm{S}_{2j-1, l}^\top(\hat{\bm{x}}_l-\bm{d}_l) +\cos\theta_j^i \sum_{1\le l\le 2k}\bm{S}_{2j, l}^\top(\hat{\bm{x}}_l-\bm{d}_l)$ with equal probability, where $\hat{\bm{x}}:=(\hat{q}_1, \hat{p}_1, \dots, \hat{q}_k, \hat{p}_k)^\top$,
 phase $\theta_j^i$ is randomly chosen from $[0,\frac{\pi}{2})$ independently. Denote the measurement outcome by $f_j^i$.
If $(f_j^i)^2>  d_0/2$, we set $z_j^i=1$; otherwise, we set $z_j^i=0$. After all the $K$ measurements, if for all~$j$, $\sum_{i=1}^{K/k} z_j^i \le R:=N\text{e}^{-c_0^2 d_0}$, then the verifier goes on to the fidelity test; otherwise, the verifiers rejects and discards all the registers left.
\item Fidelity test: randomly choose~$N-L$ registers from the remaining $N$ registers and discard them. Choose $L-m$ registers from $L$ registers. 
At $i$th register, randomly choose $j\in [k]$ and $\theta\in[0,\frac{\pi}{2})$, and measure randomly either $\hat{\tilde{q}}_j(\theta)$ or $\hat{\tilde{p}}_j(\theta)$. Denote the measurement outcomes as $\chi_i$.
After all the $L-m$ measurements, calculate 
\begin{equation} 
W^*=1+\frac{k}{2} -\frac{k}{L-m}\sum_{i=1}^{L-m} \chi_i^2.
\end{equation}
If $W^*\ge 1-\frac{\epsilon}{2m}$, the verifier accepts the~$m$ registers left as reliable copies of hypergraph states. Otherwise, the verifier rejects and discards all the remaining~$m$ registers.
\end{enumerate}

\subsubsection{Verification of hypergraph states}

Hypergraph states~\cite{Ros13,Mor17, Tak18, Zhu19a}, analogous to graph states~\cite{Rau03, Mil09}, can be described by a hypergraph, where an edge can connect more than two vertices. For a hypergraph $G=\{V, E\}$, where $V$ is the set of vertices and $E$ is the set of edges, the CV hypergraph state is
\begin{equation}\label{CVhypergraph}
\ket{G}=\prod_{e \in E}\text{e}^{-\text{i}\prod_{ i \in e} \hat{q}_i} \ket{0}_p^{\otimes n},
\end{equation}
where $\ket{0}_p$ is a momentum eigenstate with eigenvalue zero, and $\text{e}^{-\text{i}\prod_{ i \in e} \hat{q}_i}$ is a CV generalized CZ gate.
When the set $E$ of edges contains only edges connecting two vertices, $\ket{G}$ is reduced to a graph state.

As momentum eigenstates require infinite squeezing, the CV hypergraph state in Eq.~(\ref{CVhypergraph}) does not physically exist. 
To get rid of infinite squeezing, a hypergraph state can be approximated by replacing a momentum eigenstate with a finitely squeezed state, i.e., 
\begin{equation}
\ket{G}=\prod_{e \in E}\text{e}^{-\text{i}\prod_{ i \in e} \hat{q}_i} \otimes_{j=1}^n S\ket{0}^{\otimes n},
\end{equation}
where $S:=\text{e}^{\frac{\xi}{2}(\hat{a}^{\dagger\, 2}-\hat{a}^2)}$, with $\xi>0$, is a single-mode squeezing operation in momentum. 
Hypergraph states belong to the class of quantum states generated by instantaneous quantum polynomial circuits, whose measurement outcomes in computational basis cannot be classically efficiently simulated~\cite{Bre16, Dou17,Arr17}.
CV hypergraph states are non-Gaussian states and, together with Gaussian states and Gaussian operations, can be used to realize universal quantum computing.

\begin{enumerate}
\item dimension test: divide $K$ registers evenly into $k$ groups. At $i$th register of $j$th group, 
randomly choose to measure either $s\cos\theta_j \hat{q}_j+\frac{1}{s}\sin\theta_j \left( \hat{p}_j+ \sum_{e\in E|j\in e} \prod_{l\in e\slash j} \hat{q}_l\right)$ or $-s\sin\theta_j \hat{q}_j+\frac{1}{s}\cos\theta_j \left( \hat{p}_j+ \sum_{e\in E|j\in e} \prod_{l\in e\slash j} \hat{q}_l\right)$ with equal probability, where $1/s=\text{e}^{\xi}>1$,
phases $\theta_j^i$ and $\theta_l^i$ are randomly chosen from $[0,\frac{\pi}{2})$ independently. Denote the measurement outcome by $f_j^i$.
If $(f_j^i)^2>  d_0/2$, we set $z_j^i=1$; otherwise, we set $z_j^i=0$. After all the $K$ measurements, if for all~$j$, $\sum_{i=1}^{K/k} z_j^i \le R$, then the verifier goes on to the fidelity test; otherwise, the verifiers rejects and discards all the registers left.
\item Fidelity test: randomly choose~$N-L$ registers from the remaining $N$ registers and discard them. Choose $L-m$ registers from $L$ registers. 
At $i$th register, randomly choose $j\in [k]$ and $\theta\in[0,\frac{\pi}{2})$, and measure randomly either $\hat{\tilde{q}}(\theta)$ or $\hat{\tilde{p}}(\theta)$. Denote the measurement outcomes as $\chi_i$.
After all the $L-m$ measurements, calculate 
\begin{equation} 
W^*=1+\frac{k}{2} -\frac{k}{L-m}\sum_{i=1}^{L-m} \chi_i^2.
\end{equation}
If $W^*\ge 1-\frac{\epsilon}{2m}$, the verifier accepts the~$m$ registers left as reliable copies of hypergraph states. Otherwise, the verifier rejects and discards all the remaining~$m$ registers.
\end{enumerate}

\section{Verification of CV quantum devices}
State verification tests the fidelity between a given state with a target state, which is closely related to quantum benchmarks, which typically tests the fidelity between the output of a quantum device to the target output, averaged over all possible inputs.
It is shown in Ref. \cite{bai2018test} that any test of quantum devices can be realized by preparing one entangled state on the input and an ancillary system, and jointly measuring the output and the ancillary system. In case the measurement can be rephrased as a state verification task, which is true for tests of teleportation, amplification and attenuation of coherent states \cite{wu2019efficient}, the test can be implemented with state verification techniques.
The advantage of using state verification instead of measurements is that, general measurements may be hard to implement experimentally, while state verification employs only basic single-mode measurements, posing lower requirements for the verifier.

 Combining with the state verification approaches in non-i.i.d scenario, we illustrate the verification
procedure of verification of various quantum devices below. The key steps of the device-verification scheme is shown in Fig.~\ref{fig:channel}.

 \begin{figure*}
     \centering
     \includegraphics[width=0.85\textwidth]{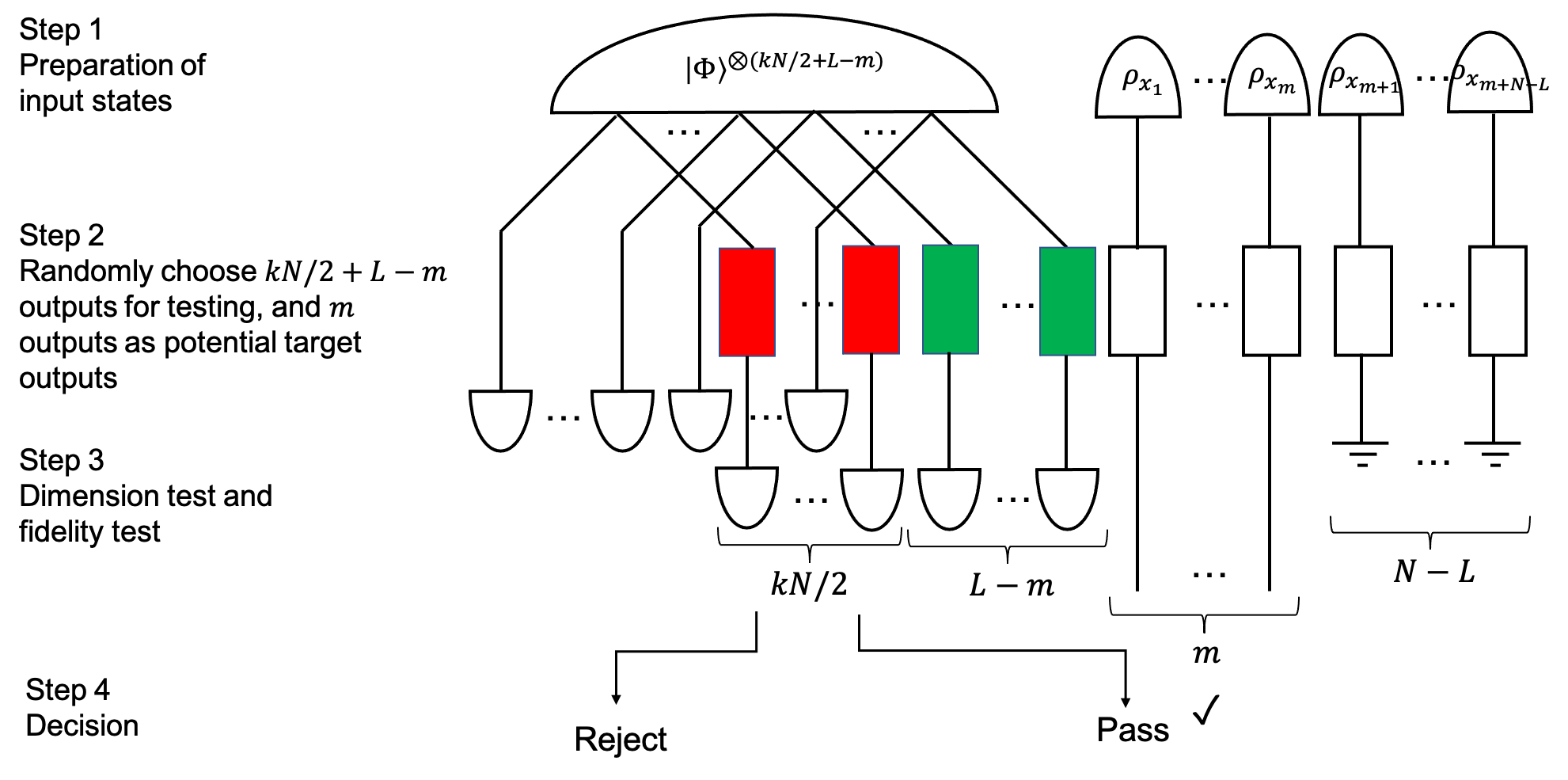}
     \caption{Begin with an unknown $k(k/2+1)N$-mode quantum channel, which is divided into $(k/2+1)N$ uses with each use, represented by a block, consisting of $k$-mode input and $k$-mode output. The verifier randomly chooses $kN/2$ uses and $L-m$ uses, represented by red blocks and green blocks, respectively. The verifier injects one part of each bipartite entangled state $\ket{\Phi}$, which is a purification of the average input state $\sum_{x\in X} p_x\rho_x$, into each input of red and green blocks. The verifier randomly chooses $m$ outputs as potential target output states with respect to any possible input states $\rho_x$. The other $N-L$ outputs are discarded. The verifier applies a dimension test at $kN/2$ random pairs of testing outputs of red blocks and the associated ancillary systems, and then a fidelity test at the other $L-m$ pairs of testing outputs of green blocks and the associated ancillary systems. If both tests are passed, then the verifier accepts the $m$ outputs as close enough target output states. Otherwise, the verifier rejects.}
     \label{fig:channel}
 \end{figure*}


\subsubsection{Verification of single-mode amplification/attenuation}
Quantum amplification protocols~\cite{Ryo08,pooser2009low,chiribella2013optimal} are important for quantum cloning and other quantum information processing protocols. Here we present a protocol to verify an optimal quantum amplifier. The input states are coherent state with Gaussian modulated amplitudes and hence, the average fidelity is
\begin{equation}\label{aveFidelityAmp}
    \bar{F}_g(\mathcal E)=\int \frac{\mathrm{d}^{2} \alpha}{\pi} \lambda \mathrm{e}^{-\lambda|\alpha|^2} 
    \bra{g\alpha} \mathcal E(\ket{\alpha}\bra{\alpha}) \ket{g \alpha},
\end{equation}
where $g$ is the amplification gain. 
An optimal amplifier can be achieved by a Gaussian amplification channel,
using two-mode squeezing when $g\ge\lambda+1$, and the maximum achievable average fidelity~(\ref{aveFidelityAmp}) is~\cite{chiribella2013optimal}
\begin{equation}
    \bar{F}_g^{\mathrm{max}}=\frac{\lambda+1}{g^2}.
\end{equation}

The verification scheme is presented in the following.
\begin{enumerate}
\item dimension test: prepare $K$ copies of  TMSV state $\ket{\kappa}_{\text{TMSV}}$. Randomly choose $K$ inputs and for each one, feed one mode of a  $\ket{\kappa}_{\text{TMSV}}$ into the channel input and keep the other mode
as a reference mode. 
For $i$th pair of output and reference, randomly choose $\theta_i\in[0,\frac{\pi}{2})$, and
 measures either 
\begin{equation}
    \hat{\tilde{q}}(\theta_i):=\cos\theta_i\left(-\sinh\kappa_0 \hat{q}_{A'}+\cosh\kappa_0 \hat{q}_{\text R}\right)+\sin\theta_i\left(\sinh\kappa_0\hat{p}_{A'}+\cosh\kappa_0\hat{p}_{\text R}\right) 
\end{equation} or 
\begin{equation}
    \hat{\tilde{p}}(\theta_i):=-\sin\theta_i\left(-\sinh\kappa_0 \hat{q}_{A'}+\cosh\kappa_0 \hat{q}_{\text R}\right)+\cos\theta_i\left(\sinh\kappa_0\hat{p}_{A'}+\cosh\kappa_0\hat{p}_{\text{R}}\right)
\end{equation}
with equal probability, where $\kappa_0=\arctanh \frac{\sqrt{\lambda+1}}{g}$.
Denote the measurement outcome by $f_i$.
If $f_i^2>  d_0/2$, we set $z_i=1$; otherwise, we set $z_i=0$. After all the $K$ rounds of measurements, if $\sum_{i=1}^{K} z_i \le R$, then the verifier goes on to the fidelity test; otherwise, the verifier aborts the test.
\item Fidelity test: prepare $L-m$ copies of  TMSV state $\ket{\kappa}_{\text{TMSV}}$. Randomly choose~$L-m$ inputs and for each one, feed one mode of a TMSV state $\ket{\kappa}_{\text{TMSV}}$ into the channel input and keep the other mode
as a reference mode. 
For $i$th pair of output and reference, randomly choose $\theta_i\in[0,\frac{\pi}{2})$, and measure either $\hat{\tilde{q}}(\theta_i)$ or $\hat{\tilde{p}}(\theta_i)$. Denote the measurement outcome as $\chi_i$.
After all the $L-m$ measurements, calculate 
\begin{equation} 
 W^*= \frac{\lambda+1}{g^2} \left[1-\frac{g^2-\lambda-1}{g^2}\frac{1}{L-m} \sum_{i=1}^{L-m} \left(\chi_i^2-\frac{1}{2}\right)\right].
\end{equation}
If $W^*\ge \frac{\lambda+1}{g^2}(1-\frac{\epsilon}{2m})$, the verifier randomly chooses $m$ outputs from the remaining ones and accepts these outputs as reliable target output states. Otherwise, the verifier rejects.
\end{enumerate}

As storage and attenuation of coherent states, the maximal achievable average fidelity is one. We have different expressions for $\hat{\tilde{q}}$ and $\hat{\tilde{p}}$
\begin{equation}
    \hat{\tilde{q}}(\theta_i):=\cos\theta_i\left(\cosh\kappa_1 \hat{q}_{A'}-\sinh\kappa_1 \hat{q}_{\text R}\right)+\sin\theta_i\left(\cosh\kappa_1\hat{p}_{A'}+\sinh\kappa_1\hat{p}_{\text R}\right)
\end{equation}
and 
\begin{equation}
    \hat{\tilde{p}}(\theta_i):=-\sin\theta_i\left(\cosh\kappa_1 \hat{q}_{A'}-\sinh\kappa_1 \hat{q}_{\text R}\right)+\cos\theta_i\left(\cosh\kappa_1\hat{p}_{A'}+\sinh\kappa_1\hat{p}_{\text R}\right),
\end{equation}
where $\kappa_1=\arctanh\frac{g}{\sqrt{\lambda+1}}$. Meanwhile, the fidelity witness estimation becomes
\begin{equation}
    W^*=1-\frac{\lambda+1-g^2}{\lambda+1}\frac{1}{L-m} \sum_{i=1}^{L-m} \left(\chi_i^2-\frac{1}{2}\right).
\end{equation}
All the other steps in verification of storage and attenuation follows directly the verification protocol for amplifiers. 
 
 \subsubsection{Verification of single-mode purification}
 
In realistic implementation, a coherent state is easily affected by Gaussian noise, leading to a thermal state $\rho_{\alpha,\mu}=\int\frac{\text{d}^2\beta}{\pi}\mu\text{e}^{-\mu|\beta|^2} \ket{\alpha+\beta}\bra{\alpha+\beta}$. With this thermal state as input, a purification protocol attempts to output  a coherent state $\ket{g\alpha}$, where $g$ is a fixed constant. Typically, the constant $g$ is smaller than 1, meaning that the amount of signal is reduced in order to achieve higher purity~\cite{andersen2005,marek2007,zhao2017}. It is also possible to consider  purification protocols that aim at amplifying the signal, that is, protocols with $g>1$~\cite{zhao2017}.    The figure of merit  for a general  purification device is
 \begin{equation}
     \bar{F}_p(\mathcal E)=\int \frac{\mathrm{d}^{2} \alpha}{\pi} \int \frac{\text{d}^2 \beta}{\pi}\lambda \mathrm{e}^{-\lambda|\alpha|^2} \mu \text{e}^{-\mu|\beta|^2}
    \bra{g\alpha} \mathcal E(\ket{\alpha+\beta}\bra{\alpha+\beta}) \ket{g \alpha}.
 \end{equation}
 We have the following verification procedure
 \begin{enumerate}
\item dimension test: randomly choose $K$ inputs and for each one, feed one mode of a TMSV state $\ket{\zeta}_{\text{TMSV}}$ into the channel input and keep the other mode
as a reference mode, where $\zeta=\arctanh\sqrt{\frac{\lambda+\mu}{\lambda+\mu+\lambda\mu}}$. 
For $i$th pair of the $K$ pairs of outputs and references, randomly choose $\theta_i\in[0,\frac{\pi}{2})$, and
when $g>\frac{\sqrt{(\lambda+\mu)(\lambda+\mu+\lambda\mu)}}{\mu}$, randomly measure either 
\begin{equation}
    \hat{\tilde{q}}(\theta_i):=\cos\theta_i\left[-\sinh\kappa_2 (\hat{q}_{A'}+\xi)+\cosh\kappa_2 \hat{q}_{\text R}\right]+\sin\theta_i\left[\sinh\kappa_2(\hat{p}_{A'}+\xi)+\cosh\kappa_2\hat{p}_{\text R}\right]
\end{equation} or 
\begin{equation}
    \hat{\tilde{p}}(\theta_i):=-\sin\theta_i\left[-\sinh\kappa_2 (\hat{q}_{A'}+\xi)+\cosh\kappa_2 \hat{q}_{\text R}\right]+\cos\theta_i\left[\sinh\kappa_2(\hat{p}_{A'}+\xi)+\cosh\kappa_2\hat{p}_{\text{R}}\right]
\end{equation}
with equal probability, where $\kappa_2=\arctanh \frac{\sqrt{(\lambda+\mu)(\lambda+\mu+\lambda\mu)}}{g\mu}$ and $\xi$ is a random number following a Gaussian distribution with mean zero and variance $\frac{g^2}{2(\lambda+\mu)}$.
Denote the measurement outcome by $f_i$.
If $f_i^2>  d_0/2$, we set $z_i=1$; otherwise, we set $z_i=0$. After all the $K$ measurements, if $\sum_{i=1}^{K} z_i \le R$, then the verifier goes on to the fidelity test; otherwise, the verifier aborts the test.
\item Fidelity test: randomly choose~$L-m$ inputs and for each one, feed one mode of a TMSV state $\ket{\zeta}_{\text{TMSV}}$ into the channel input and keep the other mode
as a reference mode. 
For $i$th pair of the $K$ pairs of outputs and references, randomly choose $\theta_i\in[0,\frac{\pi}{2})$, and measure randomly either $\hat{\tilde{q}}(\theta_i)$ or $\hat{\tilde{p}}(\theta_i)$. Denote the measurement outcome as $\chi_i$.
After all the $L-m$ measurements, calculate 
\begin{equation} 
 W^*= \frac{(\lambda+\mu)(\lambda+\mu+\lambda\mu)}{g^2 \mu^2} \left[1-\frac{g^2\mu^2-(\lambda+\mu)(\lambda+\mu+\lambda\mu)}{g^2\mu^2} \frac{1}{L-m}\sum_{i=1}^{L-m} \left(\chi_i^2-\frac{1}{2}\right)\right].
\end{equation}
If $W^*\ge  \frac{(\lambda+\mu)(\lambda+\mu+\lambda\mu)}{g^2 \mu^2} (1-\frac{\epsilon}{2m})$, the verifier randomly chooses $m$ outputs from the remaining ones and accepts these outputs as reliable target outputs. Otherwise, the verifier rejects.
\end{enumerate}

When $g<\frac{\sqrt{(\lambda+\mu)(\lambda+\mu+\lambda\mu)}}{\mu}$, the expressions of $\hat{\tilde{q}}$ and $\hat{\tilde{p}}$ are
\begin{align}
    &\hat{\tilde{q}}(\theta_i):=\cos\theta_i\left[\cosh\kappa_3( \hat{q}_{A'}+\xi)-\sinh\kappa_3 \hat{q}_{\text R}\right]+\sin\theta_i\left[\cosh\kappa_3(\hat{p}_{A'}+\xi)+\sinh\kappa_3\hat{p}_{\text R}\right]\\
    &\hat{\tilde{q}}(\theta_i):=-\sin\theta_i\left[\cosh\kappa_3 (\hat{q}_{A'}+\xi)-\sinh\kappa_3 \hat{q}_{\text R}\right]+\cos\theta_i\left[\cosh\kappa_3(\hat{p}_{A'}+\xi)+\sinh\kappa_3\hat{p}_{\text R}\right],
\end{align}
where $\kappa_3=\arctanh\frac{g\mu}{\sqrt{(\lambda+\mu)(\lambda+\mu+\lambda\mu)}}$. Meanwhile, the estimation of fidelity witness becomes
\begin{equation}
    W^*=1-\left(1-\frac{g^2\mu^2}{(\lambda+\mu)(\lambda+\mu+\lambda\mu)}\right)\frac{1}{L-m}\sum_{i=1}^{L-m} \left(\chi_i^2-\frac{1}{2}\right).
\end{equation}

\subsubsection{Verification of multi-mode quantum memory}
Quantum memories storing CV entangled states are important for building quantum networks.  To verify a quantum memory is to verify whether a quantum channel performs as an identity channel over possible input states or not.
 The figure of merit in this verification scheme is the fidelity averaged over an ensemble of $k$-mode Gaussian input states $\left\{\frac{\mathrm{d}^{2k} \bm{\alpha}}{\pi^k} \lambda^k \mathrm{e}^{-\lambda|\bm{\alpha}|^2}, U_{\bm{S}, \bm{d}}\ket{\bm{\alpha}}\right\}$, i.e.
\begin{equation}\label{aveFidelityGauUni}
     \bar{F}(\mathcal E)\coloneqq\int \frac{\mathrm{d}^{2k} \bm{\alpha}}{\pi^k} \lambda^k \mathrm{e}^{-\lambda|\bm{\alpha}|^2} 
    \langle \bm{\alpha}| U_{\bm{S}, \bm{d}}^\dagger \mathcal E(U_{\bm{S}, \bm{d}}\ket{\bm{\alpha}}\bra{\bm{\alpha}}U_{\bm{S}, \bm{d}}^\dagger) U_{\bm{S}, \bm{d}}\ket{\bm{\alpha}},
\end{equation}
where
\begin{equation}
\ket{\bm{\alpha}}\coloneqq\ket{\alpha_1}\otimes\ket{\alpha_2}\otimes\cdots\otimes\ket{\alpha_k},
\; \bm{\alpha}\coloneqq(\alpha_1,\alpha_2,\dots,\alpha_k)\in\mathbb{C}^{\otimes k}
\end{equation}
is a product of~$k$ coherent states. 
The verification scheme is in the following.

\begin{enumerate}
\item dimension test: randomly choose~$K$ inputs and divide these $K$ inputs into $k$ groups. For each input, prepare a $2k$-mode entangled state
$U_{\bm{S}, \bm{d}}\otimes U_{\bm{S}, \bm{d}}\ket{\kappa}_{\text{TMSV}}^{\otimes k}$, where the first $k$ modes are halves of TMSV states, and the other
halves are the second $k$ modes. The first $k$ modes are fed into input ports of the channel, and the other $k$ modes are kept as reference modes. 
 For $i$th channel of $j$th group,
randomly choose $\theta_i^j\in[0,\frac{\pi}{2})$, and randomly measures either 
\begin{align*}
    \hat{\tilde{q}}_j(\theta_i^j):=&\cos\theta_i^j \left[\cosh\kappa \sum_{1\le l\le 2k} \bm{S}^\top_{2j-1, l}\left(\hat{\bm{x}}_{l,A'}-\bm{d}_l\right)-\sinh\kappa \sum_{1\le l\le 2k}\bm{S}^\top_{2j-1, l}\left(\hat{\bm{x}}_{l,\text{R}}-\bm{d}_l\right)\right] \\
    &+\sin\theta_i^j \left[\cosh\kappa \sum_{1\le l\le 2k}\bm{S}^\top_{2j, l} \left(\hat{\bm{x}}_{l, A'}-\bm{d}_l\right)+\sinh\kappa \sum_{1\le l\le 2k}\bm{S}^\top_{2j, l}\left(\hat{\bm{x}}_{l, \text{R}}-\bm{d}_l\right)\right]
\end{align*} or 
\begin{align*}
    \hat{\tilde{p}}_j(\theta_i^j):=&-\sin\theta_i^j \left[\cosh\kappa \sum_{1\le l\le 2k} \bm{S}^\top_{2j-1, l}\left(\hat{\bm{x}}_{l,A'}-\bm{d}_l\right)-\sinh\kappa \sum_{1\le l\le 2k}\bm{S}^\top_{2j-1, l}\left(\hat{\bm{x}}_{l,\text{R}}-\bm{d}_l\right)\right]\\
    &+\cos\theta_i^j \left[\cosh\kappa \sum_{1\le l\le 2k}\bm{S}^\top_{2j, l} \left(\hat{\bm{x}}_{l, A'}-\bm{d}_l\right)+\sinh\kappa \sum_{1\le l\le 2k}\bm{S}^\top_{2j, l}\left(\hat{\bm{x}}_{l, \text{R}}-\bm{d}_l\right)\right]
\end{align*}
with equal probability.
Denote the measurement outcome by $f_i^j$.
If $(f_i^j)^2>  d_0/2$, we set $z_i^j=1$; otherwise, we set $z_i^j=0$. After all the $K$ rounds of tests, if for all~$j$, $\sum_{i=1}^{K} z_i^j \le R$, then the verifier goes on to the fidelity test; otherwise, the verifiers aborts the test.
\item Fidelity test: randomly choose~$L-m$ inputs from the remaining $N$ inputs and for each of them, prepare an entangled state $U_{\bm{S}, \bm{d}}\otimes U_{\bm{S}, \bm{d}}\ket{\kappa}_{\text{TMSV}}^{\otimes k}$, with first $k$ modes fed into the channel, and the other $k$ modes taken as reference modes. 
For $i$th pair of the output and reference mode, randomly choose $\theta_i\in[0,\frac{\pi}{2})$, and measure randomly either $\hat{\tilde{q}}(\theta_i)$ or $\hat{\tilde{p}}(\theta_i)$. Denote the measurement outcome as $\chi_i$.
After all the $L-m$ measurements, calculate 
\begin{equation} 
 W^*= 1-\frac{\lambda}{\lambda+1}\frac{k}{L-m} \sum_{i=1}^{L-m} \left(\chi_i^2-\frac{1}{2}\right).
\end{equation}
If $W^*\ge 1-\frac{\epsilon}{2m}$, the verifier randomly chooses $m$ outputs from the remaining ones and accepts the~$m$ outputs as reliable target output states. Otherwise, the verifier rejects.
\end{enumerate}
 
 \subsubsection{Verification of CV generalized controlled-phase gate}
In the Heisenberg picture, a k-mode CV generalized controlled-phase gate 
 $\text{e}^{-\text{i}\prod_i \hat{q}_i} $ yields a linear transformation \begin{equation}
    \hat{\bm{x}}\rightarrow \bm{T}\hat{\bm{x}}
\end{equation} in the basis of quadrature operators, where $\bm{T}$ is the matrix denoting transformation
\begin{equation}
\forall 1\le i\le k, \, \hat{q}_i \rightarrow \hat{q}_i, \, \hat{p}_i\rightarrow \hat{p}_i+ \prod_{j\neq i}\hat{q}_{j}.
\end{equation}
 Utilizing higher-order nonlinearity, this non-Gaussian quantum gate is important to realize universal quantum computing.
 The verification scheme of a generalized controlled-phase gate with respect to an ensemble of Gaussian states $\left\{\frac{\mathrm{d}^{2k} \bm{\alpha}}{\pi^k} \lambda^k \mathrm{e}^{-\lambda|\bm{\alpha}|^2}, U_{\bm{S}, \bm{d}}\ket{\bm{\alpha}}\right\}$ is presented in the following.

\begin{enumerate}
\item dimension test: randomly choose~$K$ inputs and divide these $K$ inputs into $k$ groups. For each input, prepare a $2k$-mode entangled state
$U_{\bm{S}, \bm{d}}\otimes U_{\bm{S}, \bm{d}}\ket{\kappa}_{\text{TMSV}}^{\otimes k}$, where first $k$ modes are halves of TMSV states, and second $k$ modes are the other halves. The first $k$ modes are fed into a input of the channel, and the other $k$ modes are kept as reference modes. 
 For $i$th channel of $j$th group,
randomly choose $\theta_i^j\in[0,\frac{\pi}{2})$, and randomly measures either 
\begin{align*}
    \hat{\tilde{q}}_j(\theta_i^j):=&\cos\theta_i^j \left[\cosh\kappa \sum_{1\le l\le 2k} \bm{S}^\top_{2j-1, l}\left(\sum_k \bm{T}^{-1}_{lk}\hat{\bm{x}}_{k,A'}-\bm{d}_l\right)-\sinh\kappa \sum_{1\le l\le 2k}\bm{S}^\top_{2j-1, l}\left(\hat{\bm{x}}_{l,\text{R}}-\bm{d}_l\right)\right] \\
    &+\sin\theta_i^j \left[\cosh\kappa \sum_{1\le l\le 2k}\bm{S}^\top_{2j, l} \left(\sum_k\bm{T}^{-1}_{lk}\hat{\bm{x}}_{k, A'}-\bm{d}_l\right)+\sinh\kappa \sum_{1\le l\le 2k}\bm{S}^\top_{2j, l}\left(\hat{\bm{x}}_{l, \text{R}}-\bm{d}_l\right)\right]
\end{align*} or 
\begin{align*}
    \hat{\tilde{p}}_j(\theta_i^j):=&-\sin\theta_i^j \left[\cosh\kappa \sum_{1\le l\le 2k} \bm{S}^\top_{2j-1, l}\left(\sum_k \bm{T}^{-1}_{lk}\hat{\bm{x}}_{k,A'}-\bm{d}_l\right)-\sinh\kappa \sum_{1\le l\le 2k}\bm{S}^\top_{2j-1, l}\left(\hat{\bm{x}}_{l,\text{R}}-\bm{d}_l\right)\right]\\
    &+\cos\theta_i^j \left[\cosh\kappa \sum_{1\le l\le 2k}\bm{S}^\top_{2j, l} \left(\sum_k \bm{T}^{-1}_{lk}\hat{\bm{x}}_{k, A'}-\bm{d}_l\right)+\sinh\kappa \sum_{1\le l\le 2k}\bm{S}^\top_{2j, l}\left(\hat{\bm{x}}_{l, \text{R}}-\bm{d}_l\right)\right]
\end{align*}
with equal probability.
Denote the measurement outcome by $f_i^j$.
If $(f_i^j)^2>  d_0/2$, we set $z_i^j=1$; otherwise, we set $z_i^j=0$. After all the $K$ rounds of tests, $\sum_{i=1}^{K} z_i^j \le R$, then the verifier goes on to the fidelity test; otherwise, the verifiers aborts the test.
\item Fidelity test: randomly choose~$L-m$ inputs from the remaining $N$ inputs and for each of them, prepare an entangled state $U_{\bm{S}, \bm{d}}\otimes U_{\bm{S}, \bm{d}}\ket{\kappa}_{\text{TMSV}}^{\otimes k}$, with first $k$ modes fed into the channel, and the other $k$ modes taken as reference modes. 
For $i$th output, randomly choose $\theta_i\in[0,\frac{\pi}{2})$, and measure randomly either $\hat{\tilde{q}}(\theta_i)$ or $\hat{\tilde{p}}(\theta_i)$. Denote the measurement outcome by $\chi_i$.
After all the $L-m$ measurements, calculate 
\begin{equation} 
 W^*= 1-\frac{\lambda}{\lambda+1}\frac{k}{L-m} \sum_{i=1}^{L-m} \left(\chi_i^2-\frac{1}{2}\right).
\end{equation}
If $W^*\ge 1-\frac{\epsilon}{2m}$, the verifier randomly chooses $m$ outputs from the remaining ones and accepts the~$m$ outputs as reliable target output states. Otherwise, the verifier rejects.
\end{enumerate}

 \section{Soundness and completeness}

 In this section, we show how we obtain soundness and completeness of the verification scheme. 
 We first apply a dimension test to ensure that the remaining state falls inside an almost bounded support.
Then, by discarding a large portion of the subsystems, with a high probability the reduced state falls inside a bounded support and can be approximated by an i.i.d state using the de Finetti theorem. 
Finally, by utilizing Hoeffding's inequality, we can obtain soundness. 
 
 Denote $T_1$ as the POVM element on $\mathcal H^{\otimes k\cdot K}$ corresponding to pass of the dimension test,
$T_2$ as the POVM element on $\mathcal H^{\otimes k(L-m)}$ corresponding to pass of the fidelity test.  Without loss of generality, we assume $T_1$ is applied at the first $K$ registers, $T_2$ is applied at the first $L-m$ registers of the remaining $L$ registers
after discarding $N-L$ registers, and
the last $m$ registers are used to compare with $\ket{\psi}\bra{\psi}^{\otimes m}$.
For any permutation-invariant $\rho^{N+K}$, in order to bound
\begin{equation}\label{soundness}
\tr\left[\left( T_1\otimes T_2\otimes (\mathds{1}-\ket{\psi}\bra{\psi}^{\otimes m})\right)\tr_{N-L}\rho^{(K+N)}\right],
\end{equation}
we only need to consider the diagonal part of $\rho^{(K+N)}$ in the basis $\mathcal{\bm{B}}^{\otimes K+N}$, where
\begin{equation}\label{eq:basis}
\mathcal{\bm{B}}:=\left\{ U \otimes_{i=1}^{k}\ket{n_i}\right\},
\end{equation} 
and ignore the non-diagonal entries. 
This is because in energy and fidelity tests, by choosing a random quadrature basis in each homodyne detection, both $T_1$ and $T_2$ are diagonal in the the basis of tensor products of $\mathcal{\bm{B}}$. Furthermore, $\ket{\psi}\bra{\psi}^{\otimes m}$ is diagonal in the basis $\mathcal{\bm{B}}^{\otimes m}$. 
Hence, only the diagonal part of $\rho^{N+K}$ can affect the trace in~(\ref{soundness}). Without loss of generality, in the following, we assume $\rho^{N+K}$ is diagonal.

A finite-dimensional Hilbert space of $\mathcal H^{\otimes k}$ is
\begin{equation}\label{subspace}
\bar{\mathcal H}:=\operatorname{Span}\{U \otimes_{i=1}^k \ket{n_i}| \max_{1\le i\le k} n_i<d_0\},
\end{equation} and $P_{\bar{\mathcal H}}$ and $P_{\bar{\mathcal H}^{\otimes L}}$ are the projections onto $\bar{\mathcal H}$ and $\bar{\mathcal H}^{\otimes L}$, respectively, where $U$ is the unitary such that $\ket{\psi}=U\ket{0}^{\otimes k}$. 
$P_{\bar{\mathcal H}^{\otimes N-kQ}}^N$ is the projection onto the almost bounded subspace of $\mathcal H^{\otimes k\cdot N}$ spanned by all the vectors in $\pi(\mathcal H^{\otimes kQ}\otimes \bar{\mathcal H}^{\otimes N-kQ})$ for any permutation $\pi\in S_{N}$, where $Q\in \mathbb N^+$.
Then by plugging in the decomposition 
\begin{equation}
\rho^{(N)}= P_{\bar{\mathcal H}^{\otimes N-kQ}}^N \rho^{(N)} P_{\bar{\mathcal H}^{\otimes N-kQ}}^N + \left(\mathds{1}-P_{\bar{\mathcal H}^{\otimes N-kQ}}^N\right) \rho^{(N)} \left(\mathds{1}-P_{\bar{\mathcal H}^{\otimes N-kQ}}^N\right),
\end{equation}
and using the fact that $T_1 \le \mathds{1}$ and $T_2\otimes (\mathds{1}-\ket{\psi}\bra{\psi}^{\otimes m})  \le \mathds{1}$,
we have
\begin{align}
&\tr\left[ \left(T_1\otimes T_2\otimes \left(\mathds{1}-\ket{\psi}\bra{\psi}^{\otimes m}\right)\right)\rho^{(K+L)}\right]\\
=& \tr\left(T_1\rho^{(K)}\right)\tr\left[\left(T_2\otimes \left(\mathds{1}-\ket{\psi}\bra{\psi}^{\otimes m}\right) \right) \tr_{N-L} \rho^{(N)} \right]\\
 \label{SoundnessInequal}
\le & \tr\left(T_1\rho^{(K)}\right)\left[1-\tr\left(\rho^{(N)} P_{\bar{\mathcal H}^{\otimes N-kQ}}^N\right)\right]+ \tr\left[\left(T_2\otimes \left(\mathds{1}-\ket{\psi}\bra{\psi}^{\otimes m}\right)\right)  \tr_{N-L} \left(P_{\bar{\mathcal H}^{\otimes N-kQ}}^N \rho^{(N)} P_{\bar{\mathcal H}^{\otimes N-kQ}}^N\right) \right].
\end{align}
 The first term can be bounded by the error probability of the fidelity test given that $\rho^{(N)}$ falls inside the almost bounded support of $P_{\bar{\mathcal H}^{\otimes N-kQ}}^N$. The second term gives the joint probability that the dimension test is passed whereas
 $\rho^{(N)}$ fails to be projected onto the support of $P_{\bar{\mathcal H}^{\otimes N-kQ}}^N$.
 
Before bounding the first term in~(\ref{SoundnessInequal}), we define some operators, which are useful in Prop.~\ref{prop1}.
\begin{equation}\label{Uoperator}
V_1:=\sum_{n\ge d_0} \ket{n}\bra{n},
\end{equation}
and
\begin{equation}\label{Woperator}
U_1:=\frac{1}{2} P^{q^2 \ge d_0/2}+\frac{1}{2} P^{p^2\ge d_0/2}. 
\end{equation}
are both POVM elements on $\mathcal H$.
These two operators correspond to different types of measurements: $V_1$  corresponds to a photon number detection and $U_1$ corresponds to a homodyne detection. 
It has been shown that~\cite{Lev13}
\begin{equation}\label{UWinequality}
V_1\le 4U_1+\frac{4}{c_0\sqrt{\pi d_0}}\text{e}^{-d_0 c_0^2},
\end{equation}
where $c_0=1-\frac{1}{\sqrt{2}}$.

By applying a $k$-mode unitary conjugation on both sides of $V_1$ at $j$th mode, 
\begin{equation}
\tilde{V}_1^j:= U V_1^j U^\dagger =\sum_{n_j\ge d_0} U\ket{n_j}\bra{n_j}U^\dagger
\end{equation}
is the projection onto $\bar{\mathcal H}_j^\perp$, i.e. the complementary subspace of 
\begin{equation}\label{jthsubspace}
\bar{\mathcal H}_j:=\{U \otimes_{i=1}^k \ket{n_i}:  n_j<d_0\}.
\end{equation} 
Applying unitary conjugation on both sides of $U_1$ at $j$th mode, we obtain
\begin{equation}
\tilde{U}_1^j:=U U_1^j U^\dagger=\frac{1}{2} P^{\tilde{q}_j^2\ge d_0/2}+\frac{1}{2} P^{\tilde{p}_j^2 \ge d_0/2},
\end{equation}
where $\tilde{q}_{i}$ and $\tilde{p}_{i}$ are the eigenvalues of operators $U \hat{q}_i U^\dagger$ and $U \hat{p}_i U^\dagger$, respectively.
Using (\ref{UWinequality}), we immediately get
\begin{equation}
\tilde{V}_1\le 4\tilde{U}_1+\frac{4}{c_0\sqrt{\pi d_0}}\text{e}^{-d_0 c_0^2}.
\end{equation}

To give an upper bound of the first term in (\ref{SoundnessInequal}), we utilize the following measure concentration inequality.
\begin{lemma} \label{prop1}
  Let $\rho^{K'+N}\in \mathcal S(\mathcal H^{\otimes k(K'+N)})$ be permutation-invariant among $K'+N$ subsystems, where $N\ge 2K'$.  
 Suppose a measurement, corresponding to POVM $\left\{\tilde{U}_1^j, \mathds{1}-\tilde{U}_1^j\right\}$, is applied at each of $K'$ subsystems of $\rho^{K'+N}$ and denote the classical outcomes by $(z_1, \dots, z_{K'})$, i.e.,
\begin{equation}
z_i:=\begin{cases}
0 & \text{ for } \mathds{1}-\tilde{U}_1^j \\
1 & \text{ for } \tilde{U}_1^j.
\end{cases}
\end{equation}
  Suppose a projective measurement corresponding to POVM $\left\{\tilde{V}_1^j, \mathds{1}-\tilde{V}_1^j\right\}$ is applied at each of the remaining~$N$ subsystems, and the classical outcome is denoted by $(y_1, \dots, y_N)$, i.e.,
 \begin{equation}
y_i:=\begin{cases}
0 & \text{ for } \mathds{1}-\tilde{V}_1^j,\\
1 & \text{ for } \tilde{V}_1^j.
\end{cases}
\end{equation} 
  Then for $d_0=\frac{1}{c_0^2}\ln\frac{N}{R}$, integer $R$ and $Q$ satisfying $0<R<K'$, $N\left(\frac{4R}{K'}+ \frac{1}{c_0\sqrt{\pi d_0}}\frac{R}{N}\right)<Q<N$,
  if $\sum_{i=1}^{K'} z_i\le R$, then
\begin{equation}
\operatorname{Pr}\left( \sum_{i=1}^N y_i >Q\right)\le 4\text{e}^{-\frac{K’^2}{25(K'+1)}\left(\frac{3Q}{5N}-\frac{4R}{K'}\right)^2}.
\end{equation}
\end{lemma}

  Then, from Lemma~\ref{prop1} with $K'=K/k$, we know when $\sum_{i=1}^{K/k} z_j^i\le R$, the probability
\begin{equation}
\operatorname{Pr}\left( \sum_{i=1}^N y_j^i >Q\right)\le 4 \text{e}^{-\frac{K^2}{25k(K+k)}\left(\frac{3Q}{5N}-\frac{4kR}{K}\right)^2},
\end{equation}
which is equivalent to
\begin{equation}
\tr\left(\rho^{(N)} P_{\bar{\mathcal H}_j^{\otimes N-Q}}^N\right)\ge 1-4\text{e}^{-\frac{K^2}{25k(K+k)}\left(\frac{3Q}{5N}-\frac{4kR}{K}\right)^2}.
\end{equation}
It implies that when $\sum_{i=1}^{K/k} z_j^i\le R$, with a high probability, there are at most $Q$ subsystems in the remaining~$N$ registers falling outside $\bar{\mathcal H}_j$. If, for all~$j$, there are at most $Q$ subsystems in the remaining~$N$ registers being projected outside $\bar{\mathcal H}_j$, then, in the worst case, there are $k$ mutually disjoint groups of $Q$ subsystems, each of which fails to be projected onto $\bar{\mathcal H}_j$ for a certain $j$. It implies that there are at most $kQ$ subsystems failed to be projected onto $\bar{\mathcal H}$, because being projected outside $\bar{\mathcal H}_j$ for any~$j$  indicates being projected outside $\bar{\mathcal H}$. 
Thus, when $\sum_{i=1}^{K/k} z_j^i\le R$ for all~$j$, i.e., the dimension test is passed, we have
\begin{equation}
\tr\left(\rho^{(N)} P_{\bar{\mathcal H}^{\otimes N-kQ}}^N\right)\ge \prod_{j=1}^k \tr\left(\rho^{(N)} P_{\bar{\mathcal H}_j^{\otimes N-Q}}^N\right)\ge 1-4k\text{e}^{-\frac{K^2}{25k(K+k)}\left(\frac{3Q}{5N}-\frac{4kR}{K}\right)^2}.
\end{equation}
Hence, the joint probability of passing the dimension test whereas failing to be projected onto the support of $P_{\bar{\mathcal H}^{\otimes N-kQ}}^N$ is upper bounded by
\begin{equation}
 \tr\left(T_1\rho^{(K)}\right)\left[1-\tr\left(\rho^{(N)} P_{\bar{\mathcal H}^{\otimes N-kQ}}^N\right)\right]\le 4k\text{e}^{-\frac{K^2}{25k(K+k)}\left(\frac{3Q}{5N}-\frac{4kR}{K}\right)^2}.
\end{equation}
 By setting $K=kN/2$ and $Q=15R$, the above inequality reduces to
 \begin{equation}
     \tr\left(T_1\rho^{(K)}\right)\left[1-\tr\left(\rho^{(N)} P_{\bar{\mathcal H}^{\otimes N-kQ}}^N\right)\right]\le 4k\text{e}^{-\frac{R^2 }{50(N+2)}}.
 \end{equation}
 
To bound the second term in (\ref{SoundnessInequal}),
 we show an upper bound of $ \tr\left[T_2\otimes (\mathds{1}-\ket{\psi}\bra{\psi}^{\otimes m})\tr_{N-L}\tilde{\rho}^{(N)} \right]$, where $\tilde{\rho}^{(N)}$ is a permutation-invariant diagonal density operator on the support of $P_{\bar{\mathcal H}^{\otimes N-kQ}}^N$.
\begin{lemma}\label{deFinetti}
Let $\tilde{\rho}^{(N)}$ be a permutation-invariant state within the support of $P_{\bar{\mathcal H}^{N-kQ}}^{N}$ and diagonal in the basis of $\mathcal B^{\otimes N}$, there is a probability distribution $\mu$ on $ \mathcal S(\bar{\mathcal H})$, such that 
\begin{equation}
|| \int \text{d}\mu(\sigma) \sigma^{\otimes L}-\tilde{\rho}^{(L)} ||_{1}\le \frac{2kQL}{N}+ \frac{4L \operatorname{dim}(\bar{\mathcal H})^2}{N-kQ},
\end{equation}
where $\tilde{\rho}^{(L)}$ is the reduced state of $\tilde{\rho}^{(N)}$ at $L$ subsystems.
\end{lemma}
Hence, by using $||\rho_1-\rho_2 ||_{1}=2\max_{0\le M\le \mathds{1}} \tr\left[M(\rho_1-\rho_2)\right]$, we have
\begin{align}\notag
 &\tr\left[\left(T_2\otimes \left(\mathds{1}-\ket{\psi}\bra{\psi}^{\otimes m}\right)\right)\tilde{\rho}^{(L)} \right]\\
 =&\tr\left[\left(T_2\otimes \left(\mathds{1}-\ket{\psi}\bra{\psi}^{\otimes m}\right)\right)\left(\tilde{\rho}^{(L)}- \int \text{d}\mu(\sigma) \sigma^{\otimes L}\right) \right]
 + \tr\left[\left(T_2\otimes \left(\mathds{1}-\ket{\psi}\bra{\psi}^{\otimes m}\right)\right) \int \text{d}\mu(\sigma) \sigma^{\otimes L}\right] \\
  \label{FidelityTestsoundness}
\le& \frac{15kRL}{N}+\frac{2L k^2 d_0^2}{N-15kR}+
\int \text{d}\mu(\sigma) \tr\left[\left(T_2\otimes \left(\mathds{1}-\ket{\psi}\bra{\psi}^{\otimes m}\right) \right)\sigma^{\otimes L}  \right].
\end{align}

Now to bound the second term in (\ref{FidelityTestsoundness}) from above, we only need to give an upper bound of $ \tr\left[T_2\otimes (\mathds{1}-\ket{\psi}\bra{\psi}^{\otimes m}) \sigma^{\otimes L}  \right]$ for any
$\sigma\in \mathcal S(\bar{\mathcal H})$.
We divide the proof into two cases. 
In the case that $\tr(\sigma W)> 1-\frac{\epsilon}{m}$, $\sigma$ is close to $\ket{\psi}\bra{\psi}$, i.e., $\braket{\psi|\sigma|\psi}^m>1-\epsilon$.  Then we have
    \begin{equation}
       \tr\left[\left(T_2\otimes \left(\mathds{1}-\ket{\psi}\bra{\psi}^{\otimes m}\right) \right)\sigma^{\otimes L}  \right]
< \epsilon.
    \end{equation}
    
    In the case that $\tr(\sigma W)\le 1-\frac{\epsilon}{m}$, by noting that $\tr(W\sigma)=1+\frac{k}{2}-k\mathbb{E}(\chi^2)$,
    we have $k\mathbb{E}(\chi^2)\ge\frac{k}{2}+\frac{\epsilon}{m}$. As the fidelity test is passed if and
    only if $1+\frac{k}{2}-\frac{k}{L-m}\sum_{i=1}^{L-m}\chi_i^2\ge 1-\frac{\epsilon}{2m}$, using the Hoeffding's inequality for unbounded variable~\cite{Liu19}, we get
  \begin{align}
          \tr\left[\left(T_2\otimes \left(\mathds{1}-\ket{\psi}\bra{\psi}^{\otimes m}\right)\right) \sigma^{\otimes L}  \right] 
       \le &\operatorname{Pr}\left[\frac{1}{L-m}\sum_{i=1}^{L-m}\chi_i^2\le \mathbb{E}(\chi^2)-\frac{\epsilon}{2km}\right] \\
       \le &4 \text{e}^{-\frac{(L-m)\epsilon^2}{33\times 4k^2m^2\mathbb{E}(\chi^4)}}
  \end{align}
    where $\mathbb{E}(\chi^4)$ is the expectation value of random variable $\chi^4$ on state $\sigma$. We bound $\mathbb{E}(\chi^4)$ from above by noting that $\sigma\in \mathcal S(\bar{\mathcal H})$,
    \begin{align}
        \mathbb{E}(\chi^4)=&\sum_{j=1}^k\braket{\frac{1}{2k}\left(\hat{\tilde{q}}_j^4+\hat{\tilde{p}}_j^4\right)}\\
        \le&\sum_{j=1}^{k}\frac{1}{2k}\braket{(2\hat{\tilde{n}}_j+1)^2}\\
        \le&\frac{1}{2}(2d_0-1)^2\\
        <&2d_0^2.
    \end{align}
Hence, we have
\begin{equation}
    \tr\left[\left(T_2\otimes \left(\mathds{1}-\ket{\psi}\bra{\psi}^{\otimes m}\right)\right) \sigma^{\otimes L}  \right] 
       \le 4 \text{e}^{-\frac{(L-m)\epsilon^2}{264k^2m^2 d_0^2}}.
\end{equation}

Now let us show completeness. We prove both the lower bounds of the probabilities for a tensor product of target states to pass the dimension test 
as well as the fidelity test. 
Denote 
\begin{equation}p:= \frac{1}{2} \braket{0|P^{q^2\ge d_0/2}|0}+\frac{1}{2} \braket{0|P^{p^2 \ge d_0/2} |0} =\frac{1}{\sqrt{\pi}}\int_{x^2\ge d_0/2} \text{d}x \text{e}^{-x^2}.
\end{equation}
It is easy to see 
\begin{equation}
    p <\sqrt{\frac{2}{\pi d_0}}\text{e}^{-d_0/2}<\frac{kR}{K}.
\end{equation}
Using Chernoff bound, we obtain
\begin{equation}
\operatorname{Pr}\left(\sum_{i=1}^{K/k} z_j^i>R\right) \le \text{e}^{-\frac{K}{k} D\left(\frac{kR}{K}|| p\right)},
\end{equation}
where $D(a||p):=a\log\frac{a}{p}+(1-a)\log\frac{1-a}{1-p}$.
By noting that passing the dimension test is equivalent to $\sum_{i=1}^{K/k} z_j^i\le R$ for all~$j$, we have
\begin{equation}
\tr\left(T_1\ket{\psi}\bra{\psi}^{\otimes K}\right)\ge\prod_{j=1}^k \operatorname{Pr}\left(\sum_{i=1}^{K/k} z_j^i\le R\right)\ge 1-k \text{e}^{-\frac{K}{k} D\left(\frac{kR}{K}|| p\right)}.
\end{equation}
Regarding the fidelity test, as $\braket{\psi|W|\psi}=1$, using Hoeffding's inequality, we have
\begin{align}
\tr\left(T_2\ket{\psi}\bra{\psi}^{\otimes L-m}\right)\ge &1-\operatorname{Pr}\left[\frac{1}{L-m}\sum_{i=1}^{L-m}\chi_i^2< \mathbb{E}(\chi^2)+\frac{\epsilon}{2km}\right] 
\end{align}
where in the last inequality, we have used the fact that for state $\ket{\psi}$, $\mathbb{E}(\chi^4)\le \frac{1}{2}$.
Combining these two lower bounds together, ignoring the cross term and plugging in $K=\frac{kN}{2}$, we have
\begin{equation}
\tr\left[\left(T_1\otimes T_2\right) \ket{\psi}\bra{\psi}^{\otimes K+L-m}\right]\ge1- k \text{e}^{-\frac{N}{2} D\left(\frac{2R}{N}|| p\right)}-4\text{e}^{-\frac{(L-m)\epsilon^2}{66 k^2 m^2}}.
\end{equation}

\begin{thm}
Suppose $\ket{\psi}$ is a $k$-mode Gaussian state or hypergraph state, and $K=kN/2$. Then given any unknown state in $N+K$ registers, the verifier can apply a dimension test in $K$ subsystems and the fidelity test in the other $L-m$ subsystems, as explained above.
Denote $T$ as the POVM element on $\mathcal H^{\otimes k(K+L-m)}$ corresponding to passing the verification test.
For any state $\rho^{N+K}\in \mathcal S(\mathcal H^{\otimes k(N+K)})$, which is permutation-invariant among $N+K$ registers, the verification scheme satisfies the soundness condition
\begin{align}\notag
&\tr\left[ \left(T\otimes (\mathds{1}-\ket{\psi}\bra{\psi}^{\otimes m})\right)\rho^{(K+L)}\right]\\ \label{statesoundnessAppendix}
< &  4k\text{e}^{-\frac{R^2 }{50(N+2)}} +\frac{15kRL}{N}  
 +\frac{2L k^2 d_0^2}{N-15kR}+\max\left( 4\text{e}^{-\frac{(L-m)\epsilon^2}{264 m^2k^2 d_0^2}}, \epsilon\right),
\end{align}
where $\rho^{(K+L)}$ denotes the reduced state in a subset of $K+L$ registers.
The verification scheme also satisfies the completeness condition
\begin{equation}\label{statecompletenessAppendix}
\tr\left(T \ket{\psi}\bra{\psi}^{\otimes K+L-m}\right)> 1- k \text{e}^{-\frac{N}{2} D\left(\frac{2R}{N}|| p\right)}-4\text{e}^{-\frac{(L-m)\epsilon^2}{66k^2 m^2}}.
\end{equation}
 
\end{thm}

Now we explain how to choose those parameters in the scheme to bound the completeness and soundness inequalities from above.
To bound the first term in (\ref{statesoundnessAppendix}), we have $R>\frac{N}{R}=\text{e}^{c_0^2 d_0}$, which implies $R\gg kd_0^2$ when $d_0\gg 1$. It further 
indicates that the third term is much less than the second term in (\ref{statesoundnessAppendix}).

To make the last term in (\ref{statesoundnessAppendix}) as small as $\epsilon$, we need
\begin{equation}
L=\frac{264k^2 m^2 d_0^2\ln \frac{4}{\epsilon}}{\epsilon^2}+m.
\end{equation}
To make the second term in (\ref{statesoundnessAppendix}) as small as $\epsilon$, using the relation $\frac{R}{N}=\text{e}^{-c_0^2d_0}$, we find $d_0$ must be
large enough to satisfy
\begin{equation}\label{d0scaling}
    10k\text{e}^{-c_0^2d_0}\left(\frac{264k^2 m^2 d_0^2\ln \frac{4}{\epsilon}}{\epsilon^2}+m\right)\le \epsilon.
\end{equation}
Given $d_0$ large enough, to bound the first term in (\ref{statesoundnessAppendix}) by $\epsilon$ requires
\begin{equation}\label{Rlowbound}
R\ge \frac{50(N+2)}{R}\ln\frac{4k}{\epsilon}\sim 50\ln\frac{4k}{\epsilon}\text{e}^{c_0^2d_0}.
\end{equation}
Hence, we have already obtained
\begin{equation}
    \tr\left[ \left(T\otimes (\mathds{1}-\ket{\psi}\bra{\psi}^{\otimes m})\right)\rho^{(K+L)}\right]\le 3\epsilon.
\end{equation}
By setting $\epsilon\le \epsilon_s/3$, we obtain the minimum number of $R$, $N$ and $L$ to satisfy soundness.

For completeness, we know
\begin{equation}
    4\text{e}^{-\frac{(L-m)\epsilon^2}{66k^2 m^2}}\ll 4\text{e}^{-\frac{(L-m)\epsilon^2}{264 m^2k^2 d_0^2}}\le \epsilon,
\end{equation}
so it is the first term in (\ref{statecompletenessAppendix}) that dominates the bound.
As $p\le \frac{2R}{N}\ll 1$, we have
\begin{equation}
    D\left(\frac{2R}{N}|| p\right)\sim \frac{2R}{N}\log \frac{2R}{Np}.
\end{equation}
Hence,  
\begin{equation}
    k \text{e}^{-\frac{N}{2} D\left(\frac{2R}{N}|| p\right)}\le k\left(\frac{N}{2R} \sqrt{\frac{2}{\pi d_0}}\text{e}^{-d_0/2}\right)^R
    =k\left(\sqrt{\frac{1}{2\pi d_0}} \text{e}^{(c_0^2-1/2)d_0}\right)^R.
\end{equation}
To bound the completeness, we should have $R$ large enough to satisfy
\begin{equation}
    k\left(\sqrt{\frac{1}{2\pi d_0}} \text{e}^{(c_0^2-1/2)d_0}\right)^R\le \epsilon_c.
\end{equation}
Then we get $R=O(\frac{1}{d_0}\ln \frac{k}{\epsilon_c})$. However, soundness requires $R=O\left(\ln\frac{k}{\epsilon_s}\text{e}^{c_0^2d_0}\right)$, so the scaling of $R$ is dominated by the soundness condition. 

Setting $\epsilon_s=\epsilon_c=\epsilon$, from (\ref{d0scaling}), we get $d_0=O\left(\ln\frac{km}{\epsilon}\right)$, and $\text{e}^{c_0^2 d_0}=O\left(\frac{k^3m^2d_0^2\ln\frac{1}{\epsilon}}{\epsilon^3}\right)$. 
From (\ref{Rlowbound}) and $R/N=\text{e}^{-c_0^2 d_0}$, we get
\begin{equation}
    N=O\left(\ln\frac{k}{\epsilon}\text{e}^{2c_0^2 d_0}\right)=O\left(\frac{k^6m^4}{\epsilon^6} \text{Poly}\left(\ln \frac{k m}{\epsilon}\right)\right).
\end{equation}
Hence, the sample complexity is
\begin{equation}
    n=(2k+1)N=O\left(\frac{k^7m^4}{\epsilon^6} \text{Poly}\left(\ln \frac{k m}{\epsilon}\right)\right).
\end{equation}

\section{Fidelity witness}
Vacuum state $\ket{0}$ is the unique null eigenstate of the photon number operator $\hat{n}$. Hence, $\mathds{1}-\hat{n}$ is an observable, whose expectation value $\tr\left[(\mathds{1}-\hat{n})\rho\right]$ yields a tight lower bound of the overlap between $\rho$ and 
$\ket{0}\bra{0}$, where the equality is achieved if and only if $\rho=\ket{0}\bra{0}$. The observable $\mathds{1}-\hat{n}$ is called fidelity witness~\cite{Aol15} for vacuum state $\ket{0}$. 

This idea is generalized to states $\ket{\psi}=U\ket{0}$. $\mathds{1}-U\hat{n} U^\dagger$ is a fidelity witness for the target state $\ket{\psi}=U\ket{0}$, implying that for any state $\rho$, $\tr\left[(\mathds{1}-U\hat{n}U^\dagger)\rho\right]$ is a tight lower bound of fidelity $F(\rho, \ket{\psi}\bra{\psi})=\braket{\psi|\rho|\psi}$, where the equality is true if and only if $\rho=\ket{\psi}\bra{\psi}$.
For hypergraph states, the fidelity witness is
\begin{equation}
\mathds{1}-\sum_{j=1}^k \prod_{e \in E|j\in e}\text{e}^{-\text{i}\prod_{ i \in e} \hat{q}_i} S \hat{n}_j S^\dagger \prod_{e \in E| j\in e}\text{e}^{\text{i}\prod_{ i \in e} \hat{q}_i}.
\end{equation}
Using the transformation, which the conjugation of $\prod_{e \in E}\text{e}^{-\text{i}\prod_{ i \in e} \hat{q}_i} S^{\otimes k}$ yields in the basis of quadrature operators,
\begin{equation}
\forall 1\le i\le k, \,\hat{q}_i \rightarrow \text{e}^{-\xi}\hat{q}_i, \, \hat{p}_i\rightarrow \text{e}^{\xi}\left(\hat{p}_i+ \sum_{e\in E | i\in e}\prod_{j\in e \backslash i}\hat{q}_{j}\right),
\end{equation}
we reformulate the fidelity witness as a polynomial of quadrature operators,
\begin{equation}
\left(1+\frac{k}{2}\right)\mathds{1}-\frac{1}{2}\sum_{i=1}^k \left[s^2\hat{q}_i^2+\frac{1}{s^2}\left(\hat{p}_i+\sum_{e\in E | i\in e}\prod_{j\in e \backslash i}\hat{q}_{j}\right)^2\right].
\end{equation}
 
 Fidelity witness has been generalised to average-fidelity witness for quantum channels~\cite{wu2019efficient}.
 The average fidelity can be directly estimated by measuring an observable $O_{A'\text{R}}$ at both the channel output and the reference state, i.e., $\bar{F}=\tr(O_{A'\text{R}}\mathcal E\otimes\mathds{1}(\ket{\Phi}\bra{\Phi}))$, where $\ket{\Phi}$ is a certain entangled state with one party fed into 
 the channel and the other party kept as a reference.

 Given an index set~$X$, for each input state $\rho_x$, where $x\in X$, a target channel should output pure state $\ket{\phi}_x$. The maximum of average fidelity
$\bar{F}(\mathcal E):=\sum_{x\in X}p_x \braket{\phi_x|\mathcal E_t(\rho_x)|\phi_x}$
is achieved by the target channel~$\mathcal E_t$ with the maximum denoted by $\bar{F}_{\max}$. 
Under i.i.d assumption, any benchmark test of an channel implementation $\mathcal E$ can be reformulated into a canonical test by preparing an entangled pure state~$\ket{\Psi}_{\mathrm{AR}}$, applying~$\mathcal E$ to system~A, and applying measurements on $\mathcal E\otimes \mathcal{I} (\ket{\Psi}\bra{\Psi}_{\mathrm{AR}})$ with the observable~\cite{bai2018test}
\begin{equation}\label{benchmarkObservable}
         O_{\mathrm{A}'\mathrm{R}}=\left(\mathds{1}_{\mathrm{A}'}\otimes\rho_{\mathrm{R}}^{-\frac{1}{2}}T_{\mathrm{AR}}^\dagger\right)\Omega_{\mathrm{A'A}}^{\top_{\mathrm{A}}}\left(\mathds{1}_{\mathrm{A}'}\otimes T_{\mathrm{AR}}\rho_{\mathrm{R}}^{-\frac{1}{2}} \right),
     \end{equation}
     where 
     \begin{equation}
         \rho_{\mathrm{R}}=\tr_{\mathrm{A}}(\ket{\Psi}\bra{\Psi}_{\mathrm{AR}}), \; \rho_{\mathrm{A}}=\tr_{\mathrm{R}}(\ket{\Psi}\bra{\Psi}_{\mathrm{AR}})
     \end{equation} 
      and $T_{\mathrm{AR}}$ is a partial isometry such that 
      \begin{equation}
     T_{\mathrm{AR}}^\dagger \rho_{\mathrm{A}} T_{\mathrm{AR}}=\rho_{\mathrm{R}},
     \end{equation}
and
\begin{equation}\label{average-fidelity-based-performance-operator}
    \Omega_{\mathrm{A'A}}=\sum_{x\in X}p_x \ket{\phi_x}\bra{\phi_x}\otimes \rho_x
\end{equation}
is called a performance operator.

However, $O_{\mathrm{A}'\mathrm{R}}$ cannot be measured by local measurements. 
To simplify the measurement setting, we detect an average-fidelity witness, instead of
directly measuring $O_{\mathrm{A}'\mathrm{R}}$.
The average fidelity witness is an observable $W\le O_{A'\text{R}}$, whose mean value yields a tight lower bound of $\bar{F}$.
 
 Specifically, consider storage, amplification and attenuation of coherent states, i.e., transformation $\ket{\alpha}\rightarrow \ket{g\alpha}$ with the input ensemble of coherent states with Gaussian-modulated amplitudes $\left\{\frac{1}{\pi}\text{d}\alpha\lambda\text{e}^{-\lambda|\alpha|^2}, \ket{\alpha} \right\}_{\alpha\in\mathbb C}$. A benchmark test is to estimate the average fidelity
 \begin{equation}\label{aveFidelity}
     \bar{F}_{\mathcal E}:=\int \frac{\text{d}^2\alpha}{\pi}\lambda \text{e}^{-\lambda|\alpha|^2}\braket{g\alpha| \mathcal E(\ket{\alpha}\bra{\alpha})|g\alpha}.
 \end{equation}
 When the input and the ancillary system is in the state $\ket{\kappa}_{\text{TMSV}}$,
 it has been shown that 
if $g\le \sqrt{\lambda+1}$, then the observable to measure the average-fidelity is
\begin{equation}\label{observableIdentityAttenuation}
    O_{\mathrm{A}'\mathrm{R}}=S_{\theta}(G_{\theta}\otimes \mathds{1})S_{\theta}^\dagger ,
\end{equation}
where 
 $S_{\kappa}=\text{e}^{\frac{\kappa}{2}\left(\hat{a}_{A'}\hat{a}_{\text R}+\hat{a}_{A'}^\dagger\hat{a}_{\text R}^\dagger\right)}$ is a two-mode squeezing operation at output mode and reference mode,
\begin{equation}\label{GthetaOperator}
G_\theta =\sum_{n=0}^\infty \tanh^{2n}\theta\ket{n}\bra{n}
\end{equation}
and
\begin{equation}\label{TanhTheta}
    \theta=\arctanh\frac{g}{\sqrt{\lambda+1}};
\end{equation}
otherwise,
\begin{equation}\label{observableAmplification}
   O_{\mathrm{A}'\mathrm{R}}=\tanh^2\theta' S_{\theta'}(\mathds{1}\otimes G_{\theta'})S_{\theta'}^\dagger,
\end{equation}
where 
\begin{equation}\label{TanhThetaPrime}
    \theta'=\arctanh\frac{\sqrt{\lambda+1}}{g}.
\end{equation}

Using the fact that for
any $\theta>0$, $m\in \mathbb{N}^+$,
\begin{equation}\label{lemmaInequality}
    G_\theta^{\otimes m}
\ge \mathds{1}-\frac{\sum_{i=1}^m\hat{n}_i}{\cosh^2\theta}.
\end{equation}
we obtain an average-fidelity witness 
\begin{equation}\label{averageFidelityWitnessAmplification}
     W=\frac{\lambda+1}{g^2} \left(\mathds{1}-\frac{g^2-\lambda-1}{g^2}S_{\theta'} \hat{n}_{\text{R}} S_{\theta'}^\dagger\right),
\end{equation}
when $g>\sqrt{\lambda+1}$,
and another average-fidelity witness
\begin{equation}\label{averageFidelityWitnessAttenuation}
     W=\mathds{1}-\frac{\lambda+1-g^2}{\lambda+1}S_{\theta} \hat{n}_{\text{A}'} S_{\theta}^\dagger,
\end{equation}
when $g<\sqrt{\lambda+1}$.
 It implies that $\tr\left[W \mathcal E\otimes\mathds{1}(\ket{\kappa}\bra{\kappa}_{\text{TMSV}})\right]$ is a tight lower bound of the average fidelity (\ref{aveFidelity}), where the equality is achieved if and only if $\mathcal E$ is a target channel.

\section{Rotational symmetry of verification test}

For any operator $O=\sum_{\bm{n}_1,\bm{n}_2}c_{\bm{n}_1, \bm{n}_2}\ket{\bm{n}_1}\bra{\bm{n}_2}$ on $\mathcal H^{\otimes k}$, applying a random phase rotation at each mode yields an operator diagonal in the Fock basis.
\begin{align}
&\frac{1}{(2\pi)^k}\int_{\bm{\theta}\in [0, 2\pi)^{k}}\text{d}\bm{\theta} \text{e}^{-\text{i}\theta \hat{\bm{n}}}O \text{e}^{\text{i}\theta \hat{\bm{n}}} \\
=&\frac{1}{(2\pi)^k}\sum_{\bm{n}_1,\bm{n}_2} c_{\bm{n}_1, \bm{n}_2} \int_{\bm{\theta}\in [0, 2\pi)^{k}}\text{d}\bm{\theta} \text{e}^{-\text{i}\theta \hat{\bm{n}}}\ket{\bm{n}_1}\bra{\bm{n}_2}\text{e}^{\text{i}\theta \hat{\bm{n}}} \\
=&\frac{1}{(2\pi)^k}\sum_{\bm{n}_1,\bm{n}_2} c_{\bm{n}_1, \bm{n}_2} \int_{\bm{\theta}\in [0, 2\pi)^{k}}\text{d}\bm{\theta} \text{e}^{-\text{i}\theta(\bm{n}_1-\bm{n}_2)}\ket{\bm{n}_1}\bra{\bm{n}_2}\\
=&\sum_{\bm{n}_1,\bm{n}_2} c_{\bm{n}_1, \bm{n}_2} \delta(\bm{n}_1-\bm{n}_2)\ket{\bm{n}_1}\bra{\bm{n}_2}.
\end{align}
The verifier does not need to really apply phase rotations at each mode with a randomly modulated phase. This is because a phase rotation at any mode~$j$ yildes a
transformation
\begin{equation}
    \hat{q}_j\rightarrow \hat{q}_j(\theta_j)=\cos\theta_j\hat{q}_j+\sin\theta_j \hat{p}_j,\quad \hat{p}_j\rightarrow \hat{p}_j(\theta_j)=-\sin\theta_j\hat{q}_j+\cos\theta_j\hat{p}_j.
\end{equation}
This transformation changes the quadrature basis of homodyne detections at each mode from one of $\hat{q}_j$ and $\hat{p}_j$ to one of $\hat{j}(\theta_j)$ and $\hat{p}_j(\theta)$. Thus, a random phase rotation operation before a homodyne measurement is equivalent to a homodyne detection at a random quadrature basis. 

When $U$ is Gaussian unitary or the unitary operation generating hypergraph states, for each mode $j$, $\hat{\tilde{q}}_j=U\hat{q}_jU^\dagger$ and $\hat{\tilde{p}}_j=U\hat{p}_jU^\dagger$ are polynomials of quadrature operations, which commute with each other. For any phase $\theta_j\in [0,\frac{\pi}{2})$, $\hat{\tilde{q}}_j(\theta_j)$ or $\hat{\tilde{p}}_j(\theta_j)$ can be measured by applying local a homodyne detection at each mode and the measurement outcome is the function of classical outcomes at each mode.

\section{Proof of Lemmas}

 \subsection{Proof of Lemma~\ref{prop1}}
 Before proving Lemma~\ref{prop1}, we first present two useful lemmas.
 
\begin{lemma}\label{lemma1}
Let $z_1, z_2, \dots, z_{n+k}$ be binary random variables, which does not necessarily follows i.i.d, $\prod$ be a random set of $k$ samples, and $\bar{\prod}$ be its complementary set. Then
\begin{equation}\label{Serfling1}
\text{Pr}\left(\frac{1}{n}\sum_{i\in \bar{\prod}} z_i \ge \frac{1}{k}\sum_{j\in \prod} z_j+\delta\right) \le \text{e}^{-2\delta^2\frac{nk^2}{(n+k)(k+1)}}.
\end{equation}
and
\begin{equation}\label{Serfling2}
\text{Pr}\left(\frac{1}{n}\sum_{i\in \bar{\prod}} z_i \le \frac{1}{k}\sum_{j\in \prod} z_j-\delta\right) \le \text{e}^{-2\delta^2\frac{kn^2}{(n+k)(n+1)}}.
\end{equation}
\end{lemma}

\begin{proof}
The inequality in~(\ref{Serfling1}) has been proven in~\cite{Tom17}. We follow the same idea to prove the inequality~(\ref{Serfling2}). First by noting that the mean value $\mu =\frac{1}{n+k}\left(\sum_{j\in \prod} z_j +\sum_{j\in \bar{\prod}} z_j\right)$,
we obtain
\begin{equation}
\frac{1}{n}\sum_{i\in \bar{\prod}} z_i \le \frac{1}{k}\sum_{j\in \prod} z_j-\delta \Longleftrightarrow  \frac{1}{n}\left(\mu(n+k)-\sum_{i\in \prod} z_i\right)\le \frac{1}{k}\sum_{j\in \prod} z_j-\delta.
\end{equation}
By reformulating the inequality, we find it is further equivalent to
\begin{equation}
\frac{1}{k}\sum_{i\in\prod} z_i\ge \mu+\frac{n}{n+k}\delta.
\end{equation}
Hence using Sefling's bound~\cite{Ser74}, we have
\begin{align}
\text{Pr}\left(\frac{1}{n}\sum_{i\in \bar{\prod}} z_i \le \frac{1}{k}\sum_{j\in \prod} z_j-\delta\right)&=\text{Pr}\left(\frac{1}{k}\sum_{i\in \prod} z_i \ge \mu+\frac{n}{n+k}\delta \right)\\
& \le \text{e}^{-2k (\frac{n}{n+k}\delta)^2\frac{1}{1-\frac{k-1}{n+k}}}\\
&= \text{e}^{-2\frac{kn^2\delta}{(n+k)(n+1)}}.
\end{align}
\end{proof}

\begin{lemma}\label{lemma2}
 Let $\mathcal U=\{U^0, U^1\}$ and $\mathcal V=\{V^0, V^1\}$ be both  POVMs on $\mathcal H^{\otimes k}$, $N\ge 2K'$, and 
    $(X_1, \dots, X_{N+K'})$ be the classical outcomes of measurement $\mathcal U^{\otimes K'}\otimes \mathcal V^{\otimes N}$
    applied to any permutation-invariant state $\rho^{(N+K')}\in\mathcal S(\mathcal H^{\otimes k(N+K')})$, where $X_j$ is zero for $U^0$ and $V^0$, and $X_j$ is one for $U^1$ and $V^1$. Then, 
    \begin{equation}
       \forall \delta>0, \operatorname{Pr}\left[\frac{1}{N}\sum_{j=K'+1}^{K'+N}X_{j}> \gamma_{U^1\rightarrow V^1} \left(\frac{1}{K'}\sum_{j=1}^{K' }X_j+\delta\right)+\delta\right]
        \le 4\text{e}^{-\delta^2\frac{K'^2}{K'+1}},
    \end{equation}
    where $\gamma_{U^1\rightarrow V^1}(\delta):=\operatorname{sup}_{\rho\in \mathcal S(\mathcal H^{\otimes k})}\{\tr(\rho V^1):  \tr(\rho U^1)\le \delta\}$.
\end{lemma}
 
 \begin{proof}
 The proof simply follows the proof of Lemma.\ 1 in~\cite{Ren08arXiv} and utilizes Lemma~\ref{lemma1}. 
Tensor products of binary-valued POVMs at a permutation-invariant state yields a permutation-invariant set of binary numbers. 
We can use Lemma~\ref{lemma1} for any subset of a permutation-invariant set of binary numbers.
$(X_1,\dots, X_{K'+N/2})$ are the measurement outcomes of the measurements $\mathcal U^{K'}\otimes \mathcal V^{N/2}$ 
at $K'+N/2$ subsystems. First suppose $\mathcal U^{\otimes N/2}$ is applied at the remaining $N/2$ subsystems. From Lemma~\ref{lemma1},
we have
\begin{equation}
\text{Pr}\left(\tr(U^1\rho^{(1)}|_{X_1,\dots,X_{K'+N/2} }) \ge \frac{1}{K'}\sum_{j=1}^{K'} X_j+\delta\right) \le \text{e}^{-2\delta^2\frac{NK'^2}{(N+2K')(K'+1)}},
\end{equation} 
where $\rho^{(1)}|_{X_1,\dots,X_{K'+N/2} }$ denotes the reduced state at any one of the remaining $N/2$ registers given that the measurement outcomes of $\mathcal U^{\otimes K'}\otimes \mathcal V^{\otimes N/2}$ are $X_1, \dots, X_{K'+N/2}$.
Second, suppose $\mathcal V^{\otimes N/2}$ is applied at the remaining $N/2$ subsystems. Again, Lemma~\ref{lemma1} implies
\begin{equation}
\text{Pr}\left(\tr(V^1\rho^{(1)}|_{X_1,\dots,X_{K'+N/2} }) \le \frac{2}{N}\sum_{j=K'+1}^{K'+N/2} X_j-\delta\right) \le \text{e}^{-\delta^2\frac{N^3}{2N(N+2)}}.
\end{equation}

Using the fact that $A\Longrightarrow B$ leads to $\text{Pr}(A)\le \text{Pr}(B)$, we have
\begin{align}
&\text{Pr}\left[\frac{2}{N}\sum_{j=K'+N/2+1}^{K'+N}X_{j}> \gamma_{U^1\rightarrow V^1} \left(\frac{1}{K'}\sum_{j=1}^{K' }X_j+\delta\right)+\delta\right]\\
\le & \text{Pr}\left[\left(\tr(U^1\rho^{(1)}|_{X_1,\dots,X_{K'+N/2}} ) \ge \frac{1}{K'}\sum_{j=1}^{K'} X_j+\delta\right)  \bigvee\left(  \tr(V^1\rho^{(1)}|_{X_1,\dots,X_{K'+N/2} })\le \frac{2}{N}\sum_{j=K'+1}^{K'+N/2} X_j-\delta\right)\right] \\
\le & 2\text{e}^{-\delta^2\frac{K'^2}{K'+1}},
\end{align}
where we have used $N\ge 2K'$ in the last inequality. Hence
\begin{align}\notag
 &\operatorname{Pr}\left[\frac{1}{N}\sum_{j=K'+1}^{K'+N}X_{j}> \gamma_{U^1\rightarrow V^1} \left(\frac{1}{K'}\sum_{j=1}^{K' }X_j+\delta\right)+\delta\right]\\\notag
 \le &\operatorname{Pr}\left[\left(\frac{2}{N}\sum_{j=K'+1}^{K'+N/2}X_{j}> \gamma_{U^1\rightarrow V^1} \left(\frac{1}{K'}\sum_{j=1}^{K' }X_j+\delta\right)+\delta\right)\bigvee \left(\frac{2}{N}\sum_{j=K'+1}^{K'+N/2}X_{j}> \gamma_{U^1\rightarrow V^1} \left(\frac{1}{K'}\sum_{j=1}^{K' }X_j+\delta\right)+\delta\right)\right]\\
 \le & 4\text{e}^{-\delta^2\frac{K'^2}{K'+1}}.
\end{align}
 \end{proof}
    
Now we provide the proof of Lemma~\ref{prop1}.
\begin{proof}
The inequality~(\ref{UWinequality}) implies that for any $\delta\in [0,1)$,
\begin{equation}\label{UWprobinequality}
\gamma_{\tilde{U}_1\rightarrow \tilde{V}_1}(\delta)\le 4\delta +\frac{4}{c_0\sqrt{\pi d_0}}\text{e}^{-d_0 c_0^2}.
\end{equation}
Using~(\ref{UWprobinequality}), we have if $\sum_{i=1}^{K'} z_i \le R$, then
\begin{equation}
\gamma_{\tilde{U}_1\rightarrow \tilde{V}_1}\left(\frac{1}{K'}\sum_{i=1}^{K'} z_i+\delta\right)\le\gamma_{\tilde{U}_1\rightarrow \tilde{V}_1}\left(\frac{R}{K'}+\delta\right)
\le 4\left(\frac{R}{K'}+\delta\right)+\frac{4}{c_0\sqrt{\pi d_0}}\text{e}^{-d_0 c_0^2}.
\end{equation}
Now we set 
\begin{equation}\label{QNratio}
\frac{Q}{N}=4\left(\frac{R}{K'}+\delta\right)+ \frac{4}{c_0\sqrt{\pi d_0}}\text{e}^{-d_0 c_0^2}+\delta.
\end{equation} 
Hence, we have
\begin{equation}
\frac{Q}{N}\ge \gamma_{\tilde{U}_1\rightarrow \tilde{V}_1}\left(\frac{1}{K'}\sum_{i=1}^{K'} z_i+\delta\right)+\delta,
\end{equation}
and 
\begin{equation}
\delta=\frac{1}{5}\left(\frac{Q}{N}-\frac{4R}{K'}-\frac{4}{c_0\sqrt{\pi d_0}}\text{e}^{-d_0 c_0^2}\right).
\end{equation}

Applying Lemma~\ref{lemma2}, we obtain
\begin{align}
&\operatorname{Pr}\left( \sum_{i=1}^N y_i >Q\right)\\
\le&\operatorname{Pr}\left[ \frac{1}{N}\sum_{i=1}^N y_i >\gamma_{\tilde{U}_1\rightarrow \tilde{V}_1}\left(\frac{1}{K'}\sum_{i=1}^{K'} z_i+\delta\right)+\delta \right]\\
\le& 4\text{e}^{-\frac{K'^2}{K'+1}\delta^2}\\
=&4 \text{e}^{-\frac{K'^2}{25(K'+1)}\left(\frac{Q}{N}-\frac{4R}{K'}-\frac{4}{c_0\sqrt{\pi d_0}}\text{e}^{-d_0 c_0^2}\right)^2}.
\end{align}
As $d_0= \frac{1}{c_0^2}\ln\left(\frac{15N}{Q}\right)$, we have
\begin{equation}
    \frac{4}{c_0\sqrt{\pi d_0}}\text{e}^{-d_0 c_0^2}\le \frac{2Q}{5N}.
\end{equation}
Hence
\begin{equation}
    \operatorname{Pr}\left( \sum_{i=1}^N y_i >Q\right)\le 4 \text{e}^{-\frac{K'^2}{25(K'+1)}\left(\frac{3Q}{5N}-\frac{4R}{K'}\right)^2}.
\end{equation}
\end{proof}

\subsection{Proof of Lemma~\ref{deFinetti}}

\begin{proof}
 From the definition of $P_{\bar{\mathcal H}^{N-kQ}}^{N}$, we know that it can be decomposed into to the sum of $P_{\bar{\mathcal H}^{\otimes L}}\otimes P_{\bar{\mathcal H}^{N-L-kQ}}^{N-L}$ and its complement. 
By defining
$Q_0:=P_{\bar{\mathcal H}}$
and $Q_j$, with $j\in \mathbb{N}^+$, as the projection onto $j$th unit vector in $\mathcal{\bm{B}}$,
$P_{\bar{\mathcal H}^{N-L-kQ}}^{N-L}$ can be further decomposed into a sum of mutually orthogonal projectors
\begin{equation}
P_{\bar{\mathcal H}^{N-L-kQ}}^{N-L}=\sum_{\bm{j} \in \bm{J}_{N-L-kQ}^{N-L}} Q_{j_1}\otimes Q_{j_2} \otimes \cdots \otimes Q_{j_{N-L}},
\end{equation}
where $\bm{J}_{N-L-kQ}^{N-L}$ is a set of $(N-L)$-tuples, consisting of nonnegative integers with exactly $N-L-kQ$ zeros. 

Denote $\bar{\rho}^{(N)}$ as the projection of $\tilde{\rho}^{(N)}$ onto the support of $P_{\bar{\mathcal H}^{\otimes L}}\otimes P_{\bar{\mathcal H}^{N-L-kQ}}^{N-L}$, and $\bar{\rho}_{\perp}^{(N)}$ as its orthogonal complement. 
As $\tilde{\rho}^{(N)}$ is a diagonal density operator on the basis of $\mathcal B$,
 we have $\tilde{\rho}^{(N)}=\bar{\rho}^{(N)}+\bar{\rho}_{\perp}^{(N)}$.
 
 The probability for each subsystem of $\tilde{\rho}^{(N)}$ to fall inside $\bar{\mathcal H}$ is at least $1-\frac{kQ}{N}$. Then
the probability for $\tilde{\rho}^{(N)}$ to be projected onto $P_{\bar{\mathcal H}^{\otimes L}}\otimes P_{\bar{\mathcal H}^{\otimes N-L-kQ}}^{N-L} $ is 
\begin{equation}
\tr\left(\bar{\rho}^{(N)} \right) \ge \left(1-\frac{kQ}{N}\right)^{L}\ge 1-\frac{kQL}{N}=1-\frac{10kRL}{N}.
\end{equation} 
Thus, $\tr(\bar{\rho}_{\perp}^{(N)}) \le \frac{kQL}{N}$.
 
 Using the decomposition of $P_{\bar{\mathcal H}^{N-L-kQ}}^{N-L}$, we have
 \begin{equation}
\bar{ \rho}^{(N)}=\sum_{\bm{j}\in \bm{J}_{N-L-kQ}^{N-L}} \rho_{\bm{j}},
\end{equation}
where $\rho_{\bm{j}}=P_{\bar{\mathcal H}^{\otimes L}}\otimes Q_{\bm{j}}\tilde{\rho}^{(N)}  P_{\bar{\mathcal H}^{\otimes L}} \otimes Q_{\bm{j}}$ with
$Q_{\bm{j}}=Q_{j_1}\otimes Q_{j_2}\otimes \cdots \otimes Q_{j_{N-L}}$.

Denote $\bar{\rho}_{\bm{j}}$ as the normalization of $\rho_{\bm{j}}$ and $q_{\bm{j}}=\tr(\rho_{\bm{j}})$.
Then 
\begin{equation}
\bar{\rho}^{(N)}=\sum_{\bm{j}\in \bm{J}_{N-L-kQ}^{N-L}} q_{\bm{j}}\bar{\rho}_{ \bm{j}}.
\end{equation}
Now we prove for each $\bm{j}$, the reduced state $\bar{\rho}_{\bm{j}}^{(L)}$ of $\bar{\rho}_{\bm{j}}$ is close to an i.i.d state.

As the $N-L$ discarded subsystems are permutation invariant, without loss of generality, we can write
$\bm{j}=(0, \dots, 0, j_{1}, \dots, j_{kQ})$, where $j_{1}, \dots, j_{kQ}\in \mathbb{N}^+$.
 We obtain
\begin{align}\notag
\bar{\rho}_{\bm{j}}=& \sigma^{N-kQ}\otimes \ket{e_{j_{1}}}\bra{e_{j_{1}}} \otimes \dots \otimes \ket{e_{j_{kQ}}}\bra{e_{j_{kQ}}},
\end{align}
where $\sigma^{N-kQ}$, defined as the normalization of $P_{\bar{\mathcal H}^{\otimes N-kQ}} \bar{\rho}^{(N-kQ)} P_{\bar{\mathcal H}^{\otimes N-kQ}}$, is a permutation-invariant density operator on $\bar{\mathcal H}^{\otimes (N-kQ)}$.
Then for each $\bm{j}$,
\begin{equation}
  \tr_{N-L} \bar{\rho}_{\bm{j}}
 =\tr_{N-L-kQ} \sigma^{N-kQ}.
 \end{equation}
 From the de Finetti theorem~\cite{Mat07}, we know there is a probability distribution $\mu$ on density matrix space $\mathcal S(\bar{\mathcal H})$ such that
\begin{equation}
||\int d\mu(\hat{\sigma})\hat{ \sigma}^{\otimes L}-\tr_{N-L} \bar{\rho}_{\bm{j}} ||_{1}\le \frac{4L \operatorname{dim}(\bar{\mathcal H})^2}{N-kQ}.
\end{equation}

Due to convexity of the norm and the fact that $\mu$ is independent from $\bm{j}$, and trace norm cannot be larger than two, we conclude
\begin{align}
||\tilde{\rho}^{(L)}-\int d\mu(\hat{\sigma})\hat{ \sigma}^{\otimes L} ||_{1}\le&
2\tr(\bar{\rho}_{\perp}^{(N)}) + 
\sum_{\bm{j}} q_{\bm{j}}  \frac{4L \operatorname{dim}(\bar{\mathcal H})^2}{N-kQ}\\
< &\frac{2kQL}{N}+\frac{4L \operatorname{dim}(\bar{\mathcal H})^2}{N-kQ}.
\end{align}

\end{proof}
\end{widetext}

\bibliography{refs}

\begin{thebibliography}{53}%
\makeatletter
\providecommand \@ifxundefined [1]{%
 \@ifx{#1\undefined}
}%
\providecommand \@ifnum [1]{%
 \ifnum #1\expandafter \@firstoftwo
 \else \expandafter \@secondoftwo
 \fi
}%
\providecommand \@ifx [1]{%
 \ifx #1\expandafter \@firstoftwo
 \else \expandafter \@secondoftwo
 \fi
}%
\providecommand \natexlab [1]{#1}%
\providecommand \enquote  [1]{``#1''}%
\providecommand \bibnamefont  [1]{#1}%
\providecommand \bibfnamefont [1]{#1}%
\providecommand \citenamefont [1]{#1}%
\providecommand \href@noop [0]{\@secondoftwo}%
\providecommand \href [0]{\begingroup \@sanitize@url \@href}%
\providecommand \@href[1]{\@@startlink{#1}\@@href}%
\providecommand \@@href[1]{\endgroup#1\@@endlink}%
\providecommand \@sanitize@url [0]{\catcode `\\12\catcode `\$12\catcode
  `\&12\catcode `\#12\catcode `\^12\catcode `\_12\catcode `\%12\relax}%
\providecommand \@@startlink[1]{}%
\providecommand \@@endlink[0]{}%
\providecommand \url  [0]{\begingroup\@sanitize@url \@url }%
\providecommand \@url [1]{\endgroup\@href {#1}{\urlprefix }}%
\providecommand \urlprefix  [0]{URL }%
\providecommand \Eprint [0]{\href }%
\providecommand \doibase [0]{http://dx.doi.org/}%
\providecommand \selectlanguage [0]{\@gobble}%
\providecommand \bibinfo  [0]{\@secondoftwo}%
\providecommand \bibfield  [0]{\@secondoftwo}%
\providecommand \translation [1]{[#1]}%
\providecommand \BibitemOpen [0]{}%
\providecommand \bibitemStop [0]{}%
\providecommand \bibitemNoStop [0]{.\EOS\space}%
\providecommand \EOS [0]{\spacefactor3000\relax}%
\providecommand \BibitemShut  [1]{\csname bibitem#1\endcsname}%
\let\auto@bib@innerbib\@empty
\bibitem [{\citenamefont {Braunstein}\ and\ \citenamefont {van
  Loock}(2005)}]{Bra05}%
  \BibitemOpen
  \bibfield  {author} {\bibinfo {author} {\bibfnamefont {S.~L.}\ \bibnamefont
  {Braunstein}}\ and\ \bibinfo {author} {\bibfnamefont {P.}~\bibnamefont {van
  Loock}},\ }\href {\doibase 10.1103/RevModPhys.77.513} {\bibfield  {journal}
  {\bibinfo  {journal} {Rev. Mod. Phys.}\ }\textbf {\bibinfo {volume} {77}},\
  \bibinfo {pages} {513} (\bibinfo {year} {2005})}\BibitemShut {NoStop}%
\bibitem [{\citenamefont {Weedbrook}\ \emph {et~al.}(2012)\citenamefont
  {Weedbrook}, \citenamefont {Pirandola}, \citenamefont {Garc\'{\i}a-Patr\'on},
  \citenamefont {Cerf}, \citenamefont {Ralph}, \citenamefont {Shapiro},\ and\
  \citenamefont {Lloyd}}]{Wee12}%
  \BibitemOpen
  \bibfield  {author} {\bibinfo {author} {\bibfnamefont {C.}~\bibnamefont
  {Weedbrook}}, \bibinfo {author} {\bibfnamefont {S.}~\bibnamefont
  {Pirandola}}, \bibinfo {author} {\bibfnamefont {R.}~\bibnamefont
  {Garc\'{\i}a-Patr\'on}}, \bibinfo {author} {\bibfnamefont {N.~J.}\
  \bibnamefont {Cerf}}, \bibinfo {author} {\bibfnamefont {T.~C.}\ \bibnamefont
  {Ralph}}, \bibinfo {author} {\bibfnamefont {J.~H.}\ \bibnamefont {Shapiro}},
  \ and\ \bibinfo {author} {\bibfnamefont {S.}~\bibnamefont {Lloyd}},\ }\href
  {\doibase 10.1103/RevModPhys.84.621} {\bibfield  {journal} {\bibinfo
  {journal} {Rev. Mod. Phys.}\ }\textbf {\bibinfo {volume} {84}},\ \bibinfo
  {pages} {621} (\bibinfo {year} {2012})}\BibitemShut {NoStop}%
\bibitem [{\citenamefont {Eisert}\ \emph {et~al.}(2020)\citenamefont {Eisert},
  \citenamefont {Hangleiter}, \citenamefont {Walk}, \citenamefont {Roth},
  \citenamefont {Markham}, \citenamefont {Parekh}, \citenamefont {Chabaud},\
  and\ \citenamefont {Kashefi}}]{Eis20}%
  \BibitemOpen
  \bibfield  {author} {\bibinfo {author} {\bibfnamefont {J.}~\bibnamefont
  {Eisert}}, \bibinfo {author} {\bibfnamefont {D.}~\bibnamefont {Hangleiter}},
  \bibinfo {author} {\bibfnamefont {N.}~\bibnamefont {Walk}}, \bibinfo {author}
  {\bibfnamefont {I.}~\bibnamefont {Roth}}, \bibinfo {author} {\bibfnamefont
  {D.}~\bibnamefont {Markham}}, \bibinfo {author} {\bibfnamefont
  {R.}~\bibnamefont {Parekh}}, \bibinfo {author} {\bibfnamefont
  {U.}~\bibnamefont {Chabaud}}, \ and\ \bibinfo {author} {\bibfnamefont
  {E.}~\bibnamefont {Kashefi}},\ }\href {\doibase 10.1038/s42254-020-0186-4}
  {\bibfield  {journal} {\bibinfo  {journal} {Nat. Rev. Phys.}\ ,\ \bibinfo
  {pages} {1}} (\bibinfo {year} {2020})}\BibitemShut {NoStop}%
\bibitem [{\citenamefont {Aolita}\ \emph {et~al.}(2015)\citenamefont {Aolita},
  \citenamefont {Gogolin}, \citenamefont {Kliesch},\ and\ \citenamefont
  {Eisert}}]{Aol15}%
  \BibitemOpen
  \bibfield  {author} {\bibinfo {author} {\bibfnamefont {L.}~\bibnamefont
  {Aolita}}, \bibinfo {author} {\bibfnamefont {C.}~\bibnamefont {Gogolin}},
  \bibinfo {author} {\bibfnamefont {M.}~\bibnamefont {Kliesch}}, \ and\
  \bibinfo {author} {\bibfnamefont {J.}~\bibnamefont {Eisert}},\ }\href
  {\doibase 10.1038/ncomms9498} {\bibfield  {journal} {\bibinfo  {journal}
  {Nat. Commun.}\ }\textbf {\bibinfo {volume} {6}},\ \bibinfo {pages} {1}
  (\bibinfo {year} {2015})}\BibitemShut {NoStop}%
\bibitem [{\citenamefont {Pallister}\ \emph {et~al.}(2018)\citenamefont
  {Pallister}, \citenamefont {Linden},\ and\ \citenamefont
  {Montanaro}}]{Pal18}%
  \BibitemOpen
  \bibfield  {author} {\bibinfo {author} {\bibfnamefont {S.}~\bibnamefont
  {Pallister}}, \bibinfo {author} {\bibfnamefont {N.}~\bibnamefont {Linden}}, \
  and\ \bibinfo {author} {\bibfnamefont {A.}~\bibnamefont {Montanaro}},\ }\href
  {\doibase 10.1103/PhysRevLett.120.170502} {\bibfield  {journal} {\bibinfo
  {journal} {Phys. Rev. Lett.}\ }\textbf {\bibinfo {volume} {120}},\ \bibinfo
  {pages} {170502} (\bibinfo {year} {2018})}\BibitemShut {NoStop}%
\bibitem [{\citenamefont {Takeuchi}\ and\ \citenamefont
  {Morimae}(2018)}]{Tak18}%
  \BibitemOpen
  \bibfield  {author} {\bibinfo {author} {\bibfnamefont {Y.}~\bibnamefont
  {Takeuchi}}\ and\ \bibinfo {author} {\bibfnamefont {T.}~\bibnamefont
  {Morimae}},\ }\href {\doibase 10.1103/PhysRevX.8.021060} {\bibfield
  {journal} {\bibinfo  {journal} {Phys. Rev. X}\ }\textbf {\bibinfo {volume}
  {8}},\ \bibinfo {pages} {021060} (\bibinfo {year} {2018})}\BibitemShut
  {NoStop}%
\bibitem [{\citenamefont {Liu}\ \emph {et~al.}(2019)\citenamefont {Liu},
  \citenamefont {Demarie}, \citenamefont {Tan}, \citenamefont {Aolita},\ and\
  \citenamefont {Fitzsimons}}]{Liu19}%
  \BibitemOpen
  \bibfield  {author} {\bibinfo {author} {\bibfnamefont {N.}~\bibnamefont
  {Liu}}, \bibinfo {author} {\bibfnamefont {T.~F.}\ \bibnamefont {Demarie}},
  \bibinfo {author} {\bibfnamefont {S.-H.}\ \bibnamefont {Tan}}, \bibinfo
  {author} {\bibfnamefont {L.}~\bibnamefont {Aolita}}, \ and\ \bibinfo {author}
  {\bibfnamefont {J.~F.}\ \bibnamefont {Fitzsimons}},\ }\href {\doibase
  10.1103/PhysRevA.100.062309} {\bibfield  {journal} {\bibinfo  {journal}
  {Phys. Rev. A}\ }\textbf {\bibinfo {volume} {100}},\ \bibinfo {pages}
  {062309} (\bibinfo {year} {2019})}\BibitemShut {NoStop}%
\bibitem [{\citenamefont {Zhu}\ and\ \citenamefont
  {Hayashi}(2019{\natexlab{a}})}]{Zhu19}%
  \BibitemOpen
  \bibfield  {author} {\bibinfo {author} {\bibfnamefont {H.}~\bibnamefont
  {Zhu}}\ and\ \bibinfo {author} {\bibfnamefont {M.}~\bibnamefont {Hayashi}},\
  }\href {\doibase 10.1103/PhysRevLett.123.260504} {\bibfield  {journal}
  {\bibinfo  {journal} {Phys. Rev. Lett.}\ }\textbf {\bibinfo {volume} {123}},\
  \bibinfo {pages} {260504} (\bibinfo {year} {2019}{\natexlab{a}})}\BibitemShut
  {NoStop}%
\bibitem [{\citenamefont {Zhu}\ and\ \citenamefont
  {Hayashi}(2019{\natexlab{b}})}]{Zhu19a}%
  \BibitemOpen
  \bibfield  {author} {\bibinfo {author} {\bibfnamefont {H.}~\bibnamefont
  {Zhu}}\ and\ \bibinfo {author} {\bibfnamefont {M.}~\bibnamefont {Hayashi}},\
  }\href {\doibase 10.1103/PhysRevApplied.12.054047} {\bibfield  {journal}
  {\bibinfo  {journal} {Phys. Rev. Applied}\ }\textbf {\bibinfo {volume}
  {12}},\ \bibinfo {pages} {054047} (\bibinfo {year}
  {2019}{\natexlab{b}})}\BibitemShut {NoStop}%
\bibitem [{\citenamefont {Takeuchi}\ \emph {et~al.}(2019)\citenamefont
  {Takeuchi}, \citenamefont {Mantri}, \citenamefont {Morimae}, \citenamefont
  {Mizutani},\ and\ \citenamefont {Fitzsimons}}]{Tak19}%
  \BibitemOpen
  \bibfield  {author} {\bibinfo {author} {\bibfnamefont {Y.}~\bibnamefont
  {Takeuchi}}, \bibinfo {author} {\bibfnamefont {A.}~\bibnamefont {Mantri}},
  \bibinfo {author} {\bibfnamefont {T.}~\bibnamefont {Morimae}}, \bibinfo
  {author} {\bibfnamefont {A.}~\bibnamefont {Mizutani}}, \ and\ \bibinfo
  {author} {\bibfnamefont {J.~F.}\ \bibnamefont {Fitzsimons}},\ }\href
  {\doibase 10.1038/s41534-019-0142-2} {\bibfield  {journal} {\bibinfo
  {journal} {npj Quantum Inf.}\ }\textbf {\bibinfo {volume} {5}},\ \bibinfo
  {pages} {1} (\bibinfo {year} {2019})}\BibitemShut {NoStop}%
\bibitem [{\citenamefont {Chabaud}\ \emph {et~al.}(2020)\citenamefont
  {Chabaud}, \citenamefont {Douce}, \citenamefont {Grosshans}, \citenamefont
  {Kashefi},\ and\ \citenamefont {Markham}}]{Uly19}%
  \BibitemOpen
  \bibfield  {author} {\bibinfo {author} {\bibfnamefont {U.}~\bibnamefont
  {Chabaud}}, \bibinfo {author} {\bibfnamefont {T.}~\bibnamefont {Douce}},
  \bibinfo {author} {\bibfnamefont {F.}~\bibnamefont {Grosshans}}, \bibinfo
  {author} {\bibfnamefont {E.}~\bibnamefont {Kashefi}}, \ and\ \bibinfo
  {author} {\bibfnamefont {D.}~\bibnamefont {Markham}},\ }in\ \href {\doibase
  10.4230/LIPIcs.TQC.2020.3} {\emph {\bibinfo {booktitle} {15th Conference on
  the Theory of Quantum Computation, Communication and Cryptography (TQC
  2020)}}},\ \bibinfo {series} {Leibniz International Proceedings in
  Informatics (LIPIcs)}, Vol.\ \bibinfo {volume} {158},\ \bibinfo {editor}
  {edited by\ \bibinfo {editor} {\bibfnamefont {S.~T.}\ \bibnamefont
  {Flammia}}}\ (\bibinfo  {publisher} {Schloss Dagstuhl--Leibniz-Zentrum
  f{\"u}r Informatik},\ \bibinfo {address} {Dagstuhl, Germany},\ \bibinfo
  {year} {2020})\ pp.\ \bibinfo {pages} {3:1--3:15}\BibitemShut {NoStop}%
\bibitem [{\citenamefont {Chabaud}\ \emph {et~al.}()\citenamefont {Chabaud},
  \citenamefont {Grosshans}, \citenamefont {Kashefi},\ and\ \citenamefont
  {Markham}}]{Uly20}%
  \BibitemOpen
  \bibfield  {author} {\bibinfo {author} {\bibfnamefont {U.}~\bibnamefont
  {Chabaud}}, \bibinfo {author} {\bibfnamefont {F.}~\bibnamefont {Grosshans}},
  \bibinfo {author} {\bibfnamefont {E.}~\bibnamefont {Kashefi}}, \ and\
  \bibinfo {author} {\bibfnamefont {D.}~\bibnamefont {Markham}},\ }\href@noop
  {} {\bibinfo  {journal} {arXiv:2006.03520}\ }\BibitemShut {NoStop}%
\bibitem [{\citenamefont {Wu}\ and\ \citenamefont
  {Sanders}(2019)}]{wu2019efficient}%
  \BibitemOpen
\bibfield  {journal} {  }\bibfield  {author} {\bibinfo {author} {\bibfnamefont
  {Y.-D.}\ \bibnamefont {Wu}}\ and\ \bibinfo {author} {\bibfnamefont {B.~C.}\
  \bibnamefont {Sanders}},\ }\href@noop {} {\bibfield  {journal} {\bibinfo
  {journal} {New J. Phys.}\ }\textbf {\bibinfo {volume} {21}},\ \bibinfo
  {pages} {073026} (\bibinfo {year} {2019})}\BibitemShut {NoStop}%
\bibitem [{\citenamefont {Christandl}\ and\ \citenamefont
  {Renner}(2012)}]{Chr12}%
  \BibitemOpen
  \bibfield  {author} {\bibinfo {author} {\bibfnamefont {M.}~\bibnamefont
  {Christandl}}\ and\ \bibinfo {author} {\bibfnamefont {R.}~\bibnamefont
  {Renner}},\ }\href {\doibase 10.1103/PhysRevLett.109.120403} {\bibfield
  {journal} {\bibinfo  {journal} {Phys. Rev. Lett.}\ }\textbf {\bibinfo
  {volume} {109}},\ \bibinfo {pages} {120403} (\bibinfo {year}
  {2012})}\BibitemShut {NoStop}%
\bibitem [{\citenamefont {Pfister}\ \emph {et~al.}(2018)\citenamefont
  {Pfister}, \citenamefont {Rol}, \citenamefont {Mantri}, \citenamefont
  {Tomamichel},\ and\ \citenamefont {Wehner}}]{Pfi18}%
  \BibitemOpen
  \bibfield  {author} {\bibinfo {author} {\bibfnamefont {C.}~\bibnamefont
  {Pfister}}, \bibinfo {author} {\bibfnamefont {M.~A.}\ \bibnamefont {Rol}},
  \bibinfo {author} {\bibfnamefont {A.}~\bibnamefont {Mantri}}, \bibinfo
  {author} {\bibfnamefont {M.}~\bibnamefont {Tomamichel}}, \ and\ \bibinfo
  {author} {\bibfnamefont {S.}~\bibnamefont {Wehner}},\ }\href {\doibase
  10.1038/s41467-017-00961-2} {\bibfield  {journal} {\bibinfo  {journal} {Nat.
  Commun.}\ }\textbf {\bibinfo {volume} {9}},\ \bibinfo {pages} {1} (\bibinfo
  {year} {2018})}\BibitemShut {NoStop}%
\bibitem [{\citenamefont {Wallman}(2018)}]{wallman18}%
  \BibitemOpen
  \bibfield  {author} {\bibinfo {author} {\bibfnamefont {J.~J.}\ \bibnamefont
  {Wallman}},\ }\href@noop {} {\bibfield  {journal} {\bibinfo  {journal}
  {Quantum}\ }\textbf {\bibinfo {volume} {2}},\ \bibinfo {pages} {47} (\bibinfo
  {year} {2018})}\BibitemShut {NoStop}%
\bibitem [{\citenamefont {Erhard}\ \emph {et~al.}(2019)\citenamefont {Erhard},
  \citenamefont {Wallman}, \citenamefont {Postler}, \citenamefont {Meth},
  \citenamefont {Stricker}, \citenamefont {Martinez}, \citenamefont
  {Schindler}, \citenamefont {Monz}, \citenamefont {Emerson},\ and\
  \citenamefont {Blatt}}]{erhard19}%
  \BibitemOpen
  \bibfield  {author} {\bibinfo {author} {\bibfnamefont {A.}~\bibnamefont
  {Erhard}}, \bibinfo {author} {\bibfnamefont {J.~J.}\ \bibnamefont {Wallman}},
  \bibinfo {author} {\bibfnamefont {L.}~\bibnamefont {Postler}}, \bibinfo
  {author} {\bibfnamefont {M.}~\bibnamefont {Meth}}, \bibinfo {author}
  {\bibfnamefont {R.}~\bibnamefont {Stricker}}, \bibinfo {author}
  {\bibfnamefont {E.~A.}\ \bibnamefont {Martinez}}, \bibinfo {author}
  {\bibfnamefont {P.}~\bibnamefont {Schindler}}, \bibinfo {author}
  {\bibfnamefont {T.}~\bibnamefont {Monz}}, \bibinfo {author} {\bibfnamefont
  {J.}~\bibnamefont {Emerson}}, \ and\ \bibinfo {author} {\bibfnamefont
  {R.}~\bibnamefont {Blatt}},\ }\href@noop {} {\bibfield  {journal} {\bibinfo
  {journal} {Nat. Commun.}\ }\textbf {\bibinfo {volume} {10}},\ \bibinfo
  {pages} {1} (\bibinfo {year} {2019})}\BibitemShut {NoStop}%
\bibitem [{\citenamefont {Gheorghiu}\ \emph {et~al.}(2019)\citenamefont
  {Gheorghiu}, \citenamefont {Kapourniotis},\ and\ \citenamefont
  {Kashefi}}]{Ghe19}%
  \BibitemOpen
  \bibfield  {author} {\bibinfo {author} {\bibfnamefont {A.}~\bibnamefont
  {Gheorghiu}}, \bibinfo {author} {\bibfnamefont {T.}~\bibnamefont
  {Kapourniotis}}, \ and\ \bibinfo {author} {\bibfnamefont {E.}~\bibnamefont
  {Kashefi}},\ }\href {\doibase 10.1007/s00224-018-9872-3} {\bibfield
  {journal} {\bibinfo  {journal} {Theory Comput. Syst.}\ }\textbf {\bibinfo
  {volume} {63}},\ \bibinfo {pages} {715} (\bibinfo {year} {2019})}\BibitemShut
  {NoStop}%
\bibitem [{\citenamefont {Ball}\ and\ \citenamefont
  {Banaszek}(2006)}]{ball2005hybrid}%
  \BibitemOpen
  \bibfield  {author} {\bibinfo {author} {\bibfnamefont {J.~L.}\ \bibnamefont
  {Ball}}\ and\ \bibinfo {author} {\bibfnamefont {K.}~\bibnamefont
  {Banaszek}},\ }\href@noop {} {\bibfield  {journal} {\bibinfo  {journal} {J.
  Phys. A}\ }\textbf {\bibinfo {volume} {39}},\ \bibinfo {pages} {L1} (\bibinfo
  {year} {2006})}\BibitemShut {NoStop}%
\bibitem [{\citenamefont {Christandl}\ \emph {et~al.}(2007)\citenamefont
  {Christandl}, \citenamefont {K{\"o}nig}, \citenamefont {Mitchison},\ and\
  \citenamefont {Renner}}]{Mat07}%
  \BibitemOpen
  \bibfield  {author} {\bibinfo {author} {\bibfnamefont {M.}~\bibnamefont
  {Christandl}}, \bibinfo {author} {\bibfnamefont {R.}~\bibnamefont
  {K{\"o}nig}}, \bibinfo {author} {\bibfnamefont {G.}~\bibnamefont
  {Mitchison}}, \ and\ \bibinfo {author} {\bibfnamefont {R.}~\bibnamefont
  {Renner}},\ }\href@noop {} {\bibfield  {journal} {\bibinfo  {journal}
  {Commun. Math. Phys}\ }\textbf {\bibinfo {volume} {273}},\ \bibinfo {pages}
  {473} (\bibinfo {year} {2007})}\BibitemShut {NoStop}%
\bibitem [{\citenamefont {Renner}\ and\ \citenamefont {Cirac}(2009)}]{Ren09}%
  \BibitemOpen
  \bibfield  {author} {\bibinfo {author} {\bibfnamefont {R.}~\bibnamefont
  {Renner}}\ and\ \bibinfo {author} {\bibfnamefont {J.~I.}\ \bibnamefont
  {Cirac}},\ }\href {\doibase 10.1103/PhysRevLett.102.110504} {\bibfield
  {journal} {\bibinfo  {journal} {Phys. Rev. Lett.}\ }\textbf {\bibinfo
  {volume} {102}},\ \bibinfo {pages} {110504} (\bibinfo {year}
  {2009})}\BibitemShut {NoStop}%
\bibitem [{\citenamefont {Leverrier}(2018)}]{lev18}%
  \BibitemOpen
  \bibfield  {author} {\bibinfo {author} {\bibfnamefont {A.}~\bibnamefont
  {Leverrier}},\ }\href@noop {} {\bibfield  {journal} {\bibinfo  {journal} {J.
  Math. Phys.}\ }\textbf {\bibinfo {volume} {59}},\ \bibinfo {pages} {042202}
  (\bibinfo {year} {2018})}\BibitemShut {NoStop}%
\bibitem [{\citenamefont {Leverrier}\ \emph {et~al.}(2013)\citenamefont
  {Leverrier}, \citenamefont {Garc\'{\i}a-Patr\'on}, \citenamefont {Renner},\
  and\ \citenamefont {Cerf}}]{Lev13}%
  \BibitemOpen
  \bibfield  {author} {\bibinfo {author} {\bibfnamefont {A.}~\bibnamefont
  {Leverrier}}, \bibinfo {author} {\bibfnamefont {R.}~\bibnamefont
  {Garc\'{\i}a-Patr\'on}}, \bibinfo {author} {\bibfnamefont {R.}~\bibnamefont
  {Renner}}, \ and\ \bibinfo {author} {\bibfnamefont {N.~J.}\ \bibnamefont
  {Cerf}},\ }\href {\doibase 10.1103/PhysRevLett.110.030502} {\bibfield
  {journal} {\bibinfo  {journal} {Phys. Rev. Lett.}\ }\textbf {\bibinfo
  {volume} {110}},\ \bibinfo {pages} {030502} (\bibinfo {year}
  {2013})}\BibitemShut {NoStop}%
\bibitem [{\citenamefont {Bai}\ and\ \citenamefont
  {Chiribella}(2018)}]{bai2018test}%
  \BibitemOpen
  \bibfield  {author} {\bibinfo {author} {\bibfnamefont {G.}~\bibnamefont
  {Bai}}\ and\ \bibinfo {author} {\bibfnamefont {G.}~\bibnamefont
  {Chiribella}},\ }\href@noop {} {\bibfield  {journal} {\bibinfo  {journal}
  {Phys. Rev. Lett.}\ }\textbf {\bibinfo {volume} {120}},\ \bibinfo {pages}
  {150502} (\bibinfo {year} {2018})}\BibitemShut {NoStop}%
\bibitem [{\citenamefont {Bachor}\ and\ \citenamefont
  {Ralph}(2004)}]{bachor2004}%
  \BibitemOpen
  \bibfield  {author} {\bibinfo {author} {\bibfnamefont {H.-A.}\ \bibnamefont
  {Bachor}}\ and\ \bibinfo {author} {\bibfnamefont {T.~C.}\ \bibnamefont
  {Ralph}},\ }\href@noop {} {\emph {\bibinfo {title} {A guide to experiments in
  quantum optics}}}\ (\bibinfo  {publisher} {Wiley Online Library},\ \bibinfo
  {year} {2004})\BibitemShut {NoStop}%
\bibitem [{\citenamefont {Chiribella}\ and\ \citenamefont
  {Xie}(2013)}]{chiribella2013optimal}%
  \BibitemOpen
  \bibfield  {author} {\bibinfo {author} {\bibfnamefont {G.}~\bibnamefont
  {Chiribella}}\ and\ \bibinfo {author} {\bibfnamefont {J.}~\bibnamefont
  {Xie}},\ }\href@noop {} {\bibfield  {journal} {\bibinfo  {journal} {Phys.
  Rev. Lett.}\ }\textbf {\bibinfo {volume} {110}},\ \bibinfo {pages} {213602}
  (\bibinfo {year} {2013})}\BibitemShut {NoStop}%
\bibitem [{\citenamefont {Yang}\ \emph {et~al.}(2014)\citenamefont {Yang},
  \citenamefont {Chiribella},\ and\ \citenamefont {Adesso}}]{yang14}%
  \BibitemOpen
  \bibfield  {author} {\bibinfo {author} {\bibfnamefont {Y.}~\bibnamefont
  {Yang}}, \bibinfo {author} {\bibfnamefont {G.}~\bibnamefont {Chiribella}}, \
  and\ \bibinfo {author} {\bibfnamefont {G.}~\bibnamefont {Adesso}},\ }\href
  {\doibase 10.1103/PhysRevA.90.042319} {\bibfield  {journal} {\bibinfo
  {journal} {Phys. Rev. A}\ }\textbf {\bibinfo {volume} {90}},\ \bibinfo
  {pages} {042319} (\bibinfo {year} {2014})}\BibitemShut {NoStop}%
\bibitem [{Note1()}]{Note1}%
  \BibitemOpen
  \bibinfo {note} {See the supplemental material}\BibitemShut {NoStop}%
\bibitem [{\citenamefont {Moore}(2019)}]{Moo19}%
  \BibitemOpen
  \bibfield  {author} {\bibinfo {author} {\bibfnamefont {D.~W.}\ \bibnamefont
  {Moore}},\ }\href {\doibase 10.1103/PhysRevA.100.062301} {\bibfield
  {journal} {\bibinfo  {journal} {Phys. Rev. A}\ }\textbf {\bibinfo {volume}
  {100}},\ \bibinfo {pages} {062301} (\bibinfo {year} {2019})}\BibitemShut
  {NoStop}%
\bibitem [{\citenamefont {Shaked}\ \emph {et~al.}(2018)\citenamefont {Shaked},
  \citenamefont {Michael}, \citenamefont {Vered}, \citenamefont {Bello},
  \citenamefont {Rosenbluh},\ and\ \citenamefont {Pe’er}}]{shaked2018}%
  \BibitemOpen
  \bibfield  {author} {\bibinfo {author} {\bibfnamefont {Y.}~\bibnamefont
  {Shaked}}, \bibinfo {author} {\bibfnamefont {Y.}~\bibnamefont {Michael}},
  \bibinfo {author} {\bibfnamefont {R.~Z.}\ \bibnamefont {Vered}}, \bibinfo
  {author} {\bibfnamefont {L.}~\bibnamefont {Bello}}, \bibinfo {author}
  {\bibfnamefont {M.}~\bibnamefont {Rosenbluh}}, \ and\ \bibinfo {author}
  {\bibfnamefont {A.}~\bibnamefont {Pe’er}},\ }\href@noop {} {\bibfield
  {journal} {\bibinfo  {journal} {Nat. Commun.}\ }\textbf {\bibinfo {volume}
  {9}},\ \bibinfo {pages} {1} (\bibinfo {year} {2018})}\BibitemShut {NoStop}%
\bibitem [{\citenamefont {Takeda}\ and\ \citenamefont
  {Furusawa}(2019)}]{takeda2019}%
  \BibitemOpen
  \bibfield  {author} {\bibinfo {author} {\bibfnamefont {S.}~\bibnamefont
  {Takeda}}\ and\ \bibinfo {author} {\bibfnamefont {A.}~\bibnamefont
  {Furusawa}},\ }\href@noop {} {\bibfield  {journal} {\bibinfo  {journal} {APL
  Photonics}\ }\textbf {\bibinfo {volume} {4}},\ \bibinfo {pages} {060902}
  (\bibinfo {year} {2019})}\BibitemShut {NoStop}%
\bibitem [{\citenamefont {Morimae}(2012)}]{Mor12}%
  \BibitemOpen
  \bibfield  {author} {\bibinfo {author} {\bibfnamefont {T.}~\bibnamefont
  {Morimae}},\ }\href {\doibase 10.1103/PhysRevLett.109.230502} {\bibfield
  {journal} {\bibinfo  {journal} {Phys. Rev. Lett.}\ }\textbf {\bibinfo
  {volume} {109}},\ \bibinfo {pages} {230502} (\bibinfo {year}
  {2012})}\BibitemShut {NoStop}%
\bibitem [{\citenamefont {Marshall}\ \emph {et~al.}(2016)\citenamefont
  {Marshall}, \citenamefont {Jacobsen}, \citenamefont {Sch{\"a}fermeier},
  \citenamefont {Gehring}, \citenamefont {Weedbrook},\ and\ \citenamefont
  {Andersen}}]{marshall16}%
  \BibitemOpen
  \bibfield  {author} {\bibinfo {author} {\bibfnamefont {K.}~\bibnamefont
  {Marshall}}, \bibinfo {author} {\bibfnamefont {C.~S.}\ \bibnamefont
  {Jacobsen}}, \bibinfo {author} {\bibfnamefont {C.}~\bibnamefont
  {Sch{\"a}fermeier}}, \bibinfo {author} {\bibfnamefont {T.}~\bibnamefont
  {Gehring}}, \bibinfo {author} {\bibfnamefont {C.}~\bibnamefont {Weedbrook}},
  \ and\ \bibinfo {author} {\bibfnamefont {U.~L.}\ \bibnamefont {Andersen}},\
  }\href@noop {} {\bibfield  {journal} {\bibinfo  {journal} {Nat. Commun.}\
  }\textbf {\bibinfo {volume} {7}},\ \bibinfo {pages} {1} (\bibinfo {year}
  {2016})}\BibitemShut {NoStop}%
\bibitem [{\citenamefont {Braunstein}\ \emph {et~al.}(2000)\citenamefont
  {Braunstein}, \citenamefont {Fuchs},\ and\ \citenamefont
  {Kimble}}]{braunstein2000}%
  \BibitemOpen
  \bibfield  {author} {\bibinfo {author} {\bibfnamefont {S.~L.}\ \bibnamefont
  {Braunstein}}, \bibinfo {author} {\bibfnamefont {C.~A.}\ \bibnamefont
  {Fuchs}}, \ and\ \bibinfo {author} {\bibfnamefont {H.~J.}\ \bibnamefont
  {Kimble}},\ }\href@noop {} {\bibfield  {journal} {\bibinfo  {journal} {J.
  Mod. Opt.}\ }\textbf {\bibinfo {volume} {47}},\ \bibinfo {pages} {267}
  (\bibinfo {year} {2000})}\BibitemShut {NoStop}%
\bibitem [{\citenamefont {Hammerer}\ \emph {et~al.}(2005)\citenamefont
  {Hammerer}, \citenamefont {Wolf}, \citenamefont {Polzik},\ and\ \citenamefont
  {Cirac}}]{hammerer2005}%
  \BibitemOpen
  \bibfield  {author} {\bibinfo {author} {\bibfnamefont {K.}~\bibnamefont
  {Hammerer}}, \bibinfo {author} {\bibfnamefont {M.~M.}\ \bibnamefont {Wolf}},
  \bibinfo {author} {\bibfnamefont {E.~S.}\ \bibnamefont {Polzik}}, \ and\
  \bibinfo {author} {\bibfnamefont {J.~I.}\ \bibnamefont {Cirac}},\ }\href@noop
  {} {\bibfield  {journal} {\bibinfo  {journal} {Phys. Rev. Lett.}\ }\textbf
  {\bibinfo {volume} {94}},\ \bibinfo {pages} {150503} (\bibinfo {year}
  {2005})}\BibitemShut {NoStop}%
\bibitem [{\citenamefont {Namiki}\ \emph {et~al.}(2008)\citenamefont {Namiki},
  \citenamefont {Koashi},\ and\ \citenamefont {Imoto}}]{Ryo08}%
  \BibitemOpen
  \bibfield  {author} {\bibinfo {author} {\bibfnamefont {R.}~\bibnamefont
  {Namiki}}, \bibinfo {author} {\bibfnamefont {M.}~\bibnamefont {Koashi}}, \
  and\ \bibinfo {author} {\bibfnamefont {N.}~\bibnamefont {Imoto}},\ }\href
  {\doibase 10.1103/PhysRevLett.101.100502} {\bibfield  {journal} {\bibinfo
  {journal} {Phys. Rev. Lett.}\ }\textbf {\bibinfo {volume} {101}},\ \bibinfo
  {pages} {100502} (\bibinfo {year} {2008})}\BibitemShut {NoStop}%
\bibitem [{\citenamefont {Owari}\ \emph {et~al.}(2008)\citenamefont {Owari},
  \citenamefont {Plenio}, \citenamefont {Polzik}, \citenamefont {Serafini},\
  and\ \citenamefont {Wolf}}]{owari2008}%
  \BibitemOpen
  \bibfield  {author} {\bibinfo {author} {\bibfnamefont {M.}~\bibnamefont
  {Owari}}, \bibinfo {author} {\bibfnamefont {M.~B.}\ \bibnamefont {Plenio}},
  \bibinfo {author} {\bibfnamefont {E.~S.}\ \bibnamefont {Polzik}}, \bibinfo
  {author} {\bibfnamefont {A.}~\bibnamefont {Serafini}}, \ and\ \bibinfo
  {author} {\bibfnamefont {M.~M.}\ \bibnamefont {Wolf}},\ }\href@noop {}
  {\bibfield  {journal} {\bibinfo  {journal} {New J. Phys.}\ }\textbf {\bibinfo
  {volume} {10}},\ \bibinfo {pages} {113014} (\bibinfo {year}
  {2008})}\BibitemShut {NoStop}%
\bibitem [{\citenamefont {Adesso}\ and\ \citenamefont
  {Chiribella}(2008)}]{Gerardo08}%
  \BibitemOpen
  \bibfield  {author} {\bibinfo {author} {\bibfnamefont {G.}~\bibnamefont
  {Adesso}}\ and\ \bibinfo {author} {\bibfnamefont {G.}~\bibnamefont
  {Chiribella}},\ }\href {\doibase 10.1103/PhysRevLett.100.170503} {\bibfield
  {journal} {\bibinfo  {journal} {Phys. Rev. Lett.}\ }\textbf {\bibinfo
  {volume} {100}},\ \bibinfo {pages} {170503} (\bibinfo {year}
  {2008})}\BibitemShut {NoStop}%
\bibitem [{\citenamefont {Chiribella}\ and\ \citenamefont
  {Adesso}(2014)}]{Giulio14}%
  \BibitemOpen
  \bibfield  {author} {\bibinfo {author} {\bibfnamefont {G.}~\bibnamefont
  {Chiribella}}\ and\ \bibinfo {author} {\bibfnamefont {G.}~\bibnamefont
  {Adesso}},\ }\href {\doibase 10.1103/PhysRevLett.112.010501} {\bibfield
  {journal} {\bibinfo  {journal} {Phys. Rev. Lett.}\ }\textbf {\bibinfo
  {volume} {112}},\ \bibinfo {pages} {010501} (\bibinfo {year}
  {2014})}\BibitemShut {NoStop}%
\bibitem [{\citenamefont {Rossi}\ \emph {et~al.}(2013)\citenamefont {Rossi},
  \citenamefont {Huber}, \citenamefont {Bru{\ss}},\ and\ \citenamefont
  {Macchiavello}}]{Ros13}%
  \BibitemOpen
  \bibfield  {author} {\bibinfo {author} {\bibfnamefont {M.}~\bibnamefont
  {Rossi}}, \bibinfo {author} {\bibfnamefont {M.}~\bibnamefont {Huber}},
  \bibinfo {author} {\bibfnamefont {D.}~\bibnamefont {Bru{\ss}}}, \ and\
  \bibinfo {author} {\bibfnamefont {C.}~\bibnamefont {Macchiavello}},\ }\href
  {\doibase 10.1088/1367-2630/15/11/113022} {\bibfield  {journal} {\bibinfo
  {journal} {New J. Phys.}\ }\textbf {\bibinfo {volume} {15}},\ \bibinfo
  {pages} {113022} (\bibinfo {year} {2013})}\BibitemShut {NoStop}%
\bibitem [{\citenamefont {Morimae}\ \emph {et~al.}(2017)\citenamefont
  {Morimae}, \citenamefont {Takeuchi},\ and\ \citenamefont {Hayashi}}]{Mor17}%
  \BibitemOpen
  \bibfield  {author} {\bibinfo {author} {\bibfnamefont {T.}~\bibnamefont
  {Morimae}}, \bibinfo {author} {\bibfnamefont {Y.}~\bibnamefont {Takeuchi}}, \
  and\ \bibinfo {author} {\bibfnamefont {M.}~\bibnamefont {Hayashi}},\ }\href
  {\doibase 10.1103/PhysRevA.96.062321} {\bibfield  {journal} {\bibinfo
  {journal} {Phys. Rev. A}\ }\textbf {\bibinfo {volume} {96}},\ \bibinfo
  {pages} {062321} (\bibinfo {year} {2017})}\BibitemShut {NoStop}%
\bibitem [{\citenamefont {Raussendorf}\ \emph {et~al.}(2003)\citenamefont
  {Raussendorf}, \citenamefont {Browne},\ and\ \citenamefont
  {Briegel}}]{Rau03}%
  \BibitemOpen
  \bibfield  {author} {\bibinfo {author} {\bibfnamefont {R.}~\bibnamefont
  {Raussendorf}}, \bibinfo {author} {\bibfnamefont {D.~E.}\ \bibnamefont
  {Browne}}, \ and\ \bibinfo {author} {\bibfnamefont {H.~J.}\ \bibnamefont
  {Briegel}},\ }\href {\doibase 10.1103/PhysRevA.68.022312} {\bibfield
  {journal} {\bibinfo  {journal} {Phys. Rev. A}\ }\textbf {\bibinfo {volume}
  {68}},\ \bibinfo {pages} {022312} (\bibinfo {year} {2003})}\BibitemShut
  {NoStop}%
\bibitem [{\citenamefont {Gu}\ \emph {et~al.}(2009)\citenamefont {Gu},
  \citenamefont {Weedbrook}, \citenamefont {Menicucci}, \citenamefont {Ralph},\
  and\ \citenamefont {van Loock}}]{Mil09}%
  \BibitemOpen
  \bibfield  {author} {\bibinfo {author} {\bibfnamefont {M.}~\bibnamefont
  {Gu}}, \bibinfo {author} {\bibfnamefont {C.}~\bibnamefont {Weedbrook}},
  \bibinfo {author} {\bibfnamefont {N.~C.}\ \bibnamefont {Menicucci}}, \bibinfo
  {author} {\bibfnamefont {T.~C.}\ \bibnamefont {Ralph}}, \ and\ \bibinfo
  {author} {\bibfnamefont {P.}~\bibnamefont {van Loock}},\ }\href {\doibase
  10.1103/PhysRevA.79.062318} {\bibfield  {journal} {\bibinfo  {journal} {Phys.
  Rev. A}\ }\textbf {\bibinfo {volume} {79}},\ \bibinfo {pages} {062318}
  (\bibinfo {year} {2009})}\BibitemShut {NoStop}%
\bibitem [{\citenamefont {Bremner}\ \emph {et~al.}(2016)\citenamefont
  {Bremner}, \citenamefont {Montanaro},\ and\ \citenamefont
  {Shepherd}}]{Bre16}%
  \BibitemOpen
  \bibfield  {author} {\bibinfo {author} {\bibfnamefont {M.~J.}\ \bibnamefont
  {Bremner}}, \bibinfo {author} {\bibfnamefont {A.}~\bibnamefont {Montanaro}},
  \ and\ \bibinfo {author} {\bibfnamefont {D.~J.}\ \bibnamefont {Shepherd}},\
  }\href {\doibase 10.1103/PhysRevLett.117.080501} {\bibfield  {journal}
  {\bibinfo  {journal} {Phys. Rev. Lett.}\ }\textbf {\bibinfo {volume} {117}},\
  \bibinfo {pages} {080501} (\bibinfo {year} {2016})}\BibitemShut {NoStop}%
\bibitem [{\citenamefont {Douce}\ \emph {et~al.}(2017)\citenamefont {Douce},
  \citenamefont {Markham}, \citenamefont {Kashefi}, \citenamefont {Diamanti},
  \citenamefont {Coudreau}, \citenamefont {Milman}, \citenamefont {van Loock},\
  and\ \citenamefont {Ferrini}}]{Dou17}%
  \BibitemOpen
  \bibfield  {author} {\bibinfo {author} {\bibfnamefont {T.}~\bibnamefont
  {Douce}}, \bibinfo {author} {\bibfnamefont {D.}~\bibnamefont {Markham}},
  \bibinfo {author} {\bibfnamefont {E.}~\bibnamefont {Kashefi}}, \bibinfo
  {author} {\bibfnamefont {E.}~\bibnamefont {Diamanti}}, \bibinfo {author}
  {\bibfnamefont {T.}~\bibnamefont {Coudreau}}, \bibinfo {author}
  {\bibfnamefont {P.}~\bibnamefont {Milman}}, \bibinfo {author} {\bibfnamefont
  {P.}~\bibnamefont {van Loock}}, \ and\ \bibinfo {author} {\bibfnamefont
  {G.}~\bibnamefont {Ferrini}},\ }\href {\doibase
  10.1103/PhysRevLett.118.070503} {\bibfield  {journal} {\bibinfo  {journal}
  {Phys. Rev. Lett.}\ }\textbf {\bibinfo {volume} {118}},\ \bibinfo {pages}
  {070503} (\bibinfo {year} {2017})}\BibitemShut {NoStop}%
\bibitem [{\citenamefont {Arrazola}\ \emph {et~al.}()\citenamefont {Arrazola},
  \citenamefont {Rebentrost},\ and\ \citenamefont {Weedbrook}}]{Arr17}%
  \BibitemOpen
  \bibfield  {author} {\bibinfo {author} {\bibfnamefont {J.~M.}\ \bibnamefont
  {Arrazola}}, \bibinfo {author} {\bibfnamefont {P.}~\bibnamefont
  {Rebentrost}}, \ and\ \bibinfo {author} {\bibfnamefont {C.}~\bibnamefont
  {Weedbrook}},\ }\href@noop {} {\bibinfo  {journal} {arXiv:1712.07288}\
  }\BibitemShut {NoStop}%
\bibitem [{\citenamefont {Pooser}\ \emph {et~al.}(2009)\citenamefont {Pooser},
  \citenamefont {Marino}, \citenamefont {Boyer}, \citenamefont {Jones},\ and\
  \citenamefont {Lett}}]{pooser2009low}%
  \BibitemOpen
\bibfield  {journal} {  }\bibfield  {author} {\bibinfo {author} {\bibfnamefont
  {R.~C.}\ \bibnamefont {Pooser}}, \bibinfo {author} {\bibfnamefont {A.~M.}\
  \bibnamefont {Marino}}, \bibinfo {author} {\bibfnamefont {V.}~\bibnamefont
  {Boyer}}, \bibinfo {author} {\bibfnamefont {K.~M.}\ \bibnamefont {Jones}}, \
  and\ \bibinfo {author} {\bibfnamefont {P.~D.}\ \bibnamefont {Lett}},\
  }\href@noop {} {\bibfield  {journal} {\bibinfo  {journal} {Phys. Rev. Lett.}\
  }\textbf {\bibinfo {volume} {103}},\ \bibinfo {pages} {010501} (\bibinfo
  {year} {2009})}\BibitemShut {NoStop}%
\bibitem [{\citenamefont {Andersen}\ \emph {et~al.}(2005)\citenamefont
  {Andersen}, \citenamefont {Filip}, \citenamefont {Fiur{\'a}{\v{s}}ek},
  \citenamefont {Josse},\ and\ \citenamefont {Leuchs}}]{andersen2005}%
  \BibitemOpen
  \bibfield  {author} {\bibinfo {author} {\bibfnamefont {U.~L.}\ \bibnamefont
  {Andersen}}, \bibinfo {author} {\bibfnamefont {R.}~\bibnamefont {Filip}},
  \bibinfo {author} {\bibfnamefont {J.}~\bibnamefont {Fiur{\'a}{\v{s}}ek}},
  \bibinfo {author} {\bibfnamefont {V.}~\bibnamefont {Josse}}, \ and\ \bibinfo
  {author} {\bibfnamefont {G.}~\bibnamefont {Leuchs}},\ }\href@noop {}
  {\bibfield  {journal} {\bibinfo  {journal} {Phys. Rev. A}\ }\textbf {\bibinfo
  {volume} {72}},\ \bibinfo {pages} {060301} (\bibinfo {year}
  {2005})}\BibitemShut {NoStop}%
\bibitem [{\citenamefont {Marek}\ and\ \citenamefont
  {Filip}(2007)}]{marek2007}%
  \BibitemOpen
  \bibfield  {author} {\bibinfo {author} {\bibfnamefont {P.}~\bibnamefont
  {Marek}}\ and\ \bibinfo {author} {\bibfnamefont {R.}~\bibnamefont {Filip}},\
  }\href@noop {} {\bibfield  {journal} {\bibinfo  {journal} {Quantum Inf.
  Comput.}\ }\textbf {\bibinfo {volume} {7}},\ \bibinfo {pages} {609} (\bibinfo
  {year} {2007})}\BibitemShut {NoStop}%
\bibitem [{\citenamefont {Zhao}\ and\ \citenamefont
  {Chiribella}(2017)}]{zhao2017}%
  \BibitemOpen
  \bibfield  {author} {\bibinfo {author} {\bibfnamefont {X.}~\bibnamefont
  {Zhao}}\ and\ \bibinfo {author} {\bibfnamefont {G.}~\bibnamefont
  {Chiribella}},\ }\href@noop {} {\bibfield  {journal} {\bibinfo  {journal}
  {Phy. Rev. A}\ }\textbf {\bibinfo {volume} {95}},\ \bibinfo {pages} {042303}
  (\bibinfo {year} {2017})}\BibitemShut {NoStop}%
\bibitem [{\citenamefont {Tomamichel}\ and\ \citenamefont
  {Leverrier}(2017)}]{Tom17}%
  \BibitemOpen
  \bibfield  {author} {\bibinfo {author} {\bibfnamefont {M.}~\bibnamefont
  {Tomamichel}}\ and\ \bibinfo {author} {\bibfnamefont {A.}~\bibnamefont
  {Leverrier}},\ }\href@noop {} {\bibfield  {journal} {\bibinfo  {journal}
  {Quantum}\ }\textbf {\bibinfo {volume} {1}},\ \bibinfo {pages} {14} (\bibinfo
  {year} {2017})}\BibitemShut {NoStop}%
\bibitem [{\citenamefont {Serfling}(1974)}]{Ser74}%
  \BibitemOpen
  \bibfield  {author} {\bibinfo {author} {\bibfnamefont {R.~J.}\ \bibnamefont
  {Serfling}},\ }\href@noop {} {\bibfield  {journal} {\bibinfo  {journal} {Ann.
  Stat.}\ }\textbf {\bibinfo {volume} {2}},\ \bibinfo {pages} {39} (\bibinfo
  {year} {1974})}\BibitemShut {NoStop}%
\bibitem [{\citenamefont {Renner}\ and\ \citenamefont {Cirac}()}]{Ren08arXiv}%
  \BibitemOpen
  \bibfield  {author} {\bibinfo {author} {\bibfnamefont {R.}~\bibnamefont
  {Renner}}\ and\ \bibinfo {author} {\bibfnamefont {J.~I.}\ \bibnamefont
  {Cirac}},\ }\href@noop {} {\bibinfo  {journal} {arXiv:0809.2243}\
  }\BibitemShut {NoStop}%
\end{thebibliography}%
\end{document}